\documentclass[12pt,a4paper]{amsart}
\newtheorem{theorem}{Theorem}[section]

\newtheorem{corollary}[theorem]{Corollary}
\newtheorem{proposition}[theorem]{Proposition}

\theoremstyle{definition}
\newtheorem{definition}[theorem]{Definition}

\newtheorem{example}[theorem]{Example}

\numberwithin{equation}{section}

\newcommand{\M}{\mathcal{M}}

\newcommand{\No}{\mathcal{N}}
\newcommand{\SN}{\mathcal{S}_N}
\newcommand{\Pa}{\Lambda_{N}}

\newcommand{\dist}{\mathit{dist}}
\newcommand{\Pois}{\mathit{Poisson}}
\newcommand{\Prob}{\mathrm{P}}
\newcommand{\mypath}{\mathcal{P}_k}

\newcommand{\D}{\mathbf{d}}
\newcommand{\s}{\mathbf{s}}
\newcommand{\tr}{\mathrm{tr}}
\newcommand{\reg}{\mathrm{reg}}

\newcommand{\C}{\mathbb{C}}
\newcommand{\N}{\mathbb{N}}

\newcommand{\ii}{\mathbf{i}}
\newcommand{\sgn}{\mathrm{sgn}}

\newcommand{\take}{\!\setminus\!}

\newcommand{\unsignedpermutation}{
	\pgfmatharray{\sp}{0}\let\n\pgfmathresult
	\pgfmathparse{360.0/\n}\let\segment\pgfmathresult
	\pgfmathparse{\segment/2}\let\shift\pgfmathresult
	\def\radius{1.5cm}
	\def\labelrad{1.8cm}
	\def\regionboundaryin{1.4cm}
	\def\regionboundaryout{1.6cm}
	\begin{tikzpicture}
	\foreach \x in {1,2,...,\n}
	{
		\draw[thick] (360-\x*\segment+90:\regionboundaryin)
		--(360-\x*\segment+90:\regionboundaryout);
		\pgfmatharray{\sp}{\x}\let\tmp\pgfmathresult 
		\pgfmathparse{int(\tmp)}\let\tmp\pgfmathresult 
		\node at (360-\x*\segment+90+\shift:\labelrad) {\tmp};
		\pgfmathparse{360-(\x-1)*\segment+90}\let\alpha\pgfmathresult;
		\pgfmathparse{360-(\x-1)*\segment+90-\segment}\let\beta\pgfmathresult;
		\pgfmathgreater{\tmp}{0}\let\decision\pgfmathresult
		\ifnum \decision=1
		\draw[color=black,very thick] (\alpha:\radius) arc (\alpha:\beta:\radius);
		\else
		\draw[color=gray,very thick] (\beta:\radius) arc (\beta:\alpha:\radius);
		\fi
	};
	\end{tikzpicture}
}

\usepackage{amssymb}  
\usepackage{amsmath}
\usepackage{tikz}
\usetikzlibrary{matrix,arrows}
\usepackage[normalem]{ulem}
\usepackage{enumitem}
\usepackage{hyperref}

\begin{document}

\title[Circular genome rearrangement models]{Circular genome rearrangement models: applying representation theory to evolutionary distance calculations}

\author{Venta Terauds and Jeremy Sumner}

\address{Discipline of Mathematics, School of Physical Sciences, University of Tasmania, GPO Box 252C-37, Hobart, Tasmania 7001, Australia}

\email{venta.terauds@utas.edu.au\\ jeremy.sumner@utas.edu.au  }

\thanks{This work was supported by Australian Research Council Discovery Early Career Research Award DE130100423 to JS and by use of the Nectar Research Cloud, a collaborative Australian research platform supported by the National Collaborative Research Infrastructure Strategy. 
We would like to thank Andrew Francis for helpful discussions and for providing the inspiration to follow this line of research.}

\begin{abstract}
We investigate the symmetry of circular genome rearrangement models, discuss the implementation of a new representation-theoretic method of calculating evolutionary distances between circular genomes, and give the results of some initial calculations for genomes with up to 11 regions.
\end{abstract}

%\hfill\mbox{\today}

\maketitle

\section{Introduction}

Phylogenetic modelling attempts to recover the evolutionary relationships between present-day biological organisms.
Typical input to phylogenetic methods is genomic data such as DNA or amino acid sequences.
There are many techniques available, but the classical approach is to model sequence evolution (be it DNA, amino acids, or other) as a continous-time Markov chain on a finite state space and then use likelihood (or a Bayesian approach) to estimate model parameters and the most likely evolutionary history.

Bacterial evolution is a somewhat special case in that it is usually much more dynamic than eukaryote evolution, with genome-scale events ubiquitous.
From a statistical modelling point of view, it is sensible to model bacterial evolution stochastically, using what are known as rearrangement models, and then apply a distance method. Models of genome rearrangement compare genomes with identifiably similar {\em content}, such as genes or other large scale genomic units, and focus on differences in {\em structure}, such as the order that these units appear in the genome.  

Most commonly, bacterial genomes are circular and possess a terminus (origin of replication) and antipode.
In this work, we do not consider the boundary effects of the terminus and antipode, rather assuming circular symmetry. Phylogenetic distance methods proceed by finding the evolutionary distance between pairs of the taxonomic unit of interest and then, as a graph theoretical problem, constructing evolutionary trees from these distances. For further background on bacterial genome rearrangement, rearrangement models, and distance methods, we refer the reader to \cite{darling08}, \cite{andrew14}, and \cite{phylogeny} respectively.  

Previous work on modelling bacterial evolution via rearrangement models has focused on the problem of calculating \emph{minimal distances} (see \cite{rearrangebook} for a comprehensive survey and \cite{attilaand,andrew14} for more recent work).
These are obtained by fixing an allowed set of rearrangements and asking what is the minimum number of events, chosen from this set, that is required to transform a given genome into another.
Given the combinatorial nature and consequent factorial complexity of rearrangement models, the efficient means to do this was previously only available for unrealistic choices of rearrangement (such as all inversions possible and equally likely \cite{hannenhalli1999transforming,kececioglu1995}).

In any case, the importance of modelling bacterial evolution as a stochastic process unfolding in time was recently discussed in \cite{mles}.
The overall point of that work was to show that using minimal distance as a proxy for evolutionary distance can lead to incorrect inference of evolutionary relationships and to present the first examples of maximum likelihood distances for bacteria computed using biologically reasonable models of rearrangements. 
Nonetheless, the factorial complexity of the problem remains an issue.

Concurrent work, described in \cite{jezandpet}, applied the representation theory of the symmetric group to convert this combinatorial problem into a numerical one. In particular, it was shown that the likelihood function, calculated via an infinite sum over all possible numbers of events \cite{mles}, may be written analytically as a finite sum when a Poisson distribution of events in time is assumed.
A computational advantage was further obtained in \cite{jezandpet} by decomposing into the irreducible representations of the symmetric group and thus reducing redundancy in the calculation.
Despite these improvements, the factorial complexity of the problem remains and the broader goal of that work was to initiate the introduction of numerical approximations to the problem.

The present work is intended to take the first basic steps into that domain.
Although we do not pursue the possibility of numerical approximations to a significant level, we systematically  sort through the principal practical issues that arise when attempting a concrete implementation of the techniques of \cite{jezandpet} in commonly available software (our computations were performed in SageMath \cite{sage}).

To begin, in Section~\ref{sec technique} we set out the key components of rearrangement models and the maximum likelihood approach to calculating evolutionary distance.
As foreshadowed in \cite{jezandpet}, we show that the technique easily adapts to cover models in which different rearrangements occur with different probabilities. 

Section \ref{sec symmetry} focuses on some theoretical aspects of rearrangement models in general. Symmetries arising from the circular structure of the genome itself have been discussed in detail previously (for example in \cite{andrew14, jezandpet}). However, the biological model (which we define to be the set of chosen rearrangement moves along with their probability of occurence) may also exhibit symmetry. Given our goal of estimating distance via maximum likelihood, we consider two genomes equivalent under a given model if they produce the same likelihood functions of elapsed time. We find three types of equivalences for circular genomes that are associated with different model symmetries, describe the equivalence classes, and connect each to existing combinatorial constructions and sequences in the Online Encyclopaedia of Integer Sequences (OEIS) \cite{oeis}.

Yet another aspect of symmetry is considered in Section~\ref{sec dihedral}. Here we use the dihedral symmetry of the genomes and Frobenius' character formula to identify specific irreducible representations that, irrespective of the model and the number of regions, make a contribution of zero to the likelihoood function and  thus may be ignored.

In Section~\ref{sec implement} we move to the application of the ideas thus far presented. 
Fixing two models of genome rearrangement, we give explicit, exact expressions for the likelihood functions for genomes with 5 and 6 regions. 
Similar expressions for genomes with 3 and 4 regions, under one of the models we consider, were obtained in \cite{mles} and \cite{jezandpet} respectively.

For genomes with $N\!>\!6$ regions, exact expressions are not attainable, and we explore some of the issues that arise in numerically calculating the likelihood for $N\!>\!6$ regions. 
We find, for example, that the method of calculating projection operators proposed in \cite{jezandpet} works well for genomes with 7 regions or less, but becomes numerically unstable for genomes with more regions. 
We overcome this difficulty by using an alternative definition of projection operators --- via eigenvectors --- that allows direct computation of the required ``partial traces''.

This allows us to calculate likelihood functions (numerically) and thus compute distance estimates, under each of our models, for genomes with up to 11 regions. We include a selection of results (plots of likelihood functions and maximum likelihood estimates); comprehensive results are provided in the supplementary material.

In comparison with minimal distances, our results show that the maximum likelihood approach provides a more refined estimate of time elapsed.
While minimal distances take integer values on a highly restricted range, rendering many genomes indistinguishable, our calculated maximum likelihood estimates of time elapsed are all \emph{distinct} (up to expected model symmetries).  
We also see that the maximum likelihood approach provides additional information regarding the \emph{uncertainty} of the resulting estimate of elapsed time (via the curvature of the likelihood curve around the optimum).

We conclude the paper by putting our results in a broader context and highlighting some of the many avenues that we aim to explore in forthcoming work. 
%Though not strictly necessary (all eigenvalues should be real for the model we consider), when working as much as possible with built-in capabilities of SageMath \cite{sage} we found that our approach was most stable when implemented over the complex field. Hence we present the theory this way with, for example, unitary representations favoured over orthogonal representations.

%Stuff about circular genomes and bacterial genomes. We're focusing on rearrangements. Historical stuff. Combinatorial complexity. Blah Blah. Recent work: reference Jeremy/Andrew/Peter's paper on the model \cite{jezandpet}, Andrew's paper on models \cite{crazyworld}, Jeremy/Barbara/etc on why we're using likelihood \cite{mles}. 
%
%This is what's in this paper and how we set it out. blah. 

\section{Rearrangement models}\label{sec technique}

Full details of the technique that we apply here may be found in the recent paper \cite{jezandpet}. 
As is usual, a circular genome with $N$ regions of interest is represented as a permutation $\sigma$ on the set $\{1,2,\ldots ,N\}$; that is, an element of the symmetric group $\SN$. 
Here, {\em region} refers to an identifiable segment of the genome that is common across multiple genomes (for example, but not exclusively, a gene). 
To be mathematically precise, $\sigma$ is a mapping from the $N$ regions of interest to their $N$ possible positions on the genome: $\sigma(i) \!=\! j$ means that region $i$ is in position $j$, where we have chosen an initial reference frame --- starting position and direction --- and numbered positions consecutively around the genome. More detail on this ``position paradigm''  (and others) may be found in \cite{crazyworld}. 

Note that we predominantly use {\em cycle notation} for elements of $\SN$. For example, if we write $\sigma\!=\!(i,j,k)$, then $\sigma(i)\!=\!j, \sigma(j)\!=\!k$ and $\sigma(k)\!=\!i$.

\begin{figure}
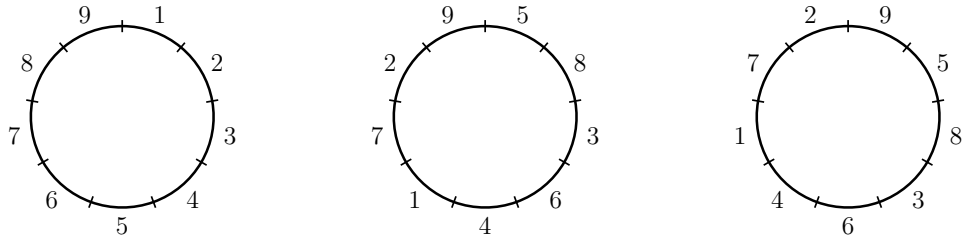

	\begin{center}
		%\hspace{-12mm}
		\resizebox{4cm}{!}{
			\def\sp{{9,1,2,3,4,5,6,7,8,9}}\unsignedpermutation
		}
		\hspace{.5cm}
		\resizebox{4cm}{!}{
			\def\sp{{9,5,8,3,6,4,1,7,2,9}}\unsignedpermutation
		}
		\hspace{.5cm}
		\resizebox{4cm}{!}{
			\def\sp{{9,9,5,8,3,6,4,1,7,2}}\unsignedpermutation
		}
\caption{The diagram on the left shows a reference frame for circular genomes with nine regions. With respect to this reference frame, the left diagram represents the identity permutation $e$ and the other two the permutations $(1,6,4,5)(2,8)$ and $(1,7,8,3,4,6,5,2,9)\in\mathcal{S}_9$. The right two diagrams clearly represent the same genome.}		
	\end{center}
\end{figure}

We model evolution of genomes through rearrangement. A rearrangement of the regions of a genome $\sigma\in\SN$ is a permutation, $a\in\SN$, that acts on $\sigma$ on the left: 
\[ \sigma\mapsto a\sigma\,.\]
As $\sigma$ is a map from regions to positions, the rearrangement $a$ maps positions to positions, so it's a permutation in the usual sense. 
Then $a\sigma(i) \!=\! j$ means that region $i$ was in position $\sigma(i)$, which was then ``rearranged'', via $a$, to position $j$. 

Given that our genomes are circular, with no distinguished regions or positions, they possess a dihedral symmetry. This means that the action of certain rearrangements, precisely those belonging to the dihedral group $D_N\subseteq\SN$, corresponds to the physical action of rotating the genome or flipping it over. These actions give rise to distinct permutations with respect to the original reference frame, but the permutations all represent the same genome. Accordingly, we make the following equivalence of permutations: for $\sigma\in\SN$,
\[ \sigma \equiv d\,\sigma \; \textrm{ for any } d \in D_N \,.\]
The set of symmetries of a genome $\sigma\in\SN$ is then the equivalence class
\[ [\sigma] := \{d\,\sigma \, :\, d \in D_N \} \,.\]
Note that we do not identify a genome with its equivalence class. That is, we think of genomes as permutations $\sigma\in\SN$ and not as cosets $D_N \sigma$ (as is done in \cite{mles}). 
Depending on our choice of reference frame and on the physical orientation in which we observe it, a genome may be represented by any element of the equivalence class $[\sigma]$. Conversely, any two elements of the class $[\sigma]$ represent the same genome. 

A biological model for evolution is chosen by firstly specifying a set of allowed rearrangements, $\M\subseteq \SN\take D_N$. Rearrangements, which we shall also refer to as rearrangement events, are most commonly inversions \cite{darling08}, in which a section (contiguous set of regions) of the genome is taken out, flipped over, and put back in.\footnote{Note that we do not include gene orientation here (this would put our genomes and rearrangements into the hyperoctahedral group), but this will be addressed in future work.} In this work, we assume that all genomes may be obtained from the reference genome via a sequence of rearrangements chosen from $\M$, that is, that the set $\M$ generates the group $\SN$. Note that this is not necessary for any of our results; actually, if $\M$ did not generate the whole of $\SN$, then this would actually simplify the situation --- that is, we would be working with a smaller group and thus a reduced combinatorial complexity. 

The next component of the biological model is a mapping $w:\M\to (0,1]$ that assigns a probability to each rearrangement event. Formally, for each $a\in\M$, 
\[  w(a) = \Prob(\textrm{rearrangement event } a)\,,\]
and $\sum_{a\in\M} w(a) \!=\! 1$. The rearrangement events are assumed to be independent and to occur randomly in time according to a given distribution, $\dist$. In practice, we shall take $\dist$ to be the Poisson distribution.
We refer to the triplet $(\M, w, \dist)$ as the {\em full biological model} and, as our preliminary calculations and results are independent of the distribution of events in time, we shall often use the tuple $(\M,w)$ and refer to it as the {\em biological model}.

%For our practical results implementing the full biological model, we shall use the Poisson distribution. 
We calculate the distance from a reference genome $\sigma_0$ to a genome $\sigma$ as the maximum likelihood estimate of the time taken for $\sigma_0$ to evolve to $\sigma$ under the full biological model (c.f. \cite{mles}). 
To this end, we consider (for each $k\in\N$) all possible ways of applying $k$ rearrangements from $\M$ to the permutation $\sigma_0$ to obtain the permutation $\sigma$, that is, the set of {\em paths}
\[ \mypath(\sigma_0\to\sigma) := \{ (a_1,a_2, \ldots,a_k)\in\M^k\,:\, a_k \ldots a_2 a_1 \sigma_0 = \sigma \}\,, \]
and thus calculate the {\em path probability}
\[ \begin{split}
		\beta_k(\sigma_0 \to \sigma) &:= \Prob(\sigma_0 \to \sigma \textrm{ via } k \textrm{ rearrangement events}) \\
									 &= \sum_{(a_1,a_2, \ldots,a_k)\in\mypath(\sigma_0\to\sigma)}\Prob((a_1,a_2, \ldots,a_k))\\
									 &= \sum_{(a_1,a_2, \ldots,a_k)\in\mypath(\sigma_0\to\sigma)} w(a_1)w(a_2) \ldots w(a_k) \,.	
\end{split} \]
As $\beta_k(\sigma_0 \!\to\! \sigma) \!=\! \beta_k(e \!\to\! \sigma \sigma_0^{-1})$ for any $\sigma,\sigma_0\in\SN$ (where $e$ is the identity permutation), we can, without loss of generality, take our reference genome to be the identity. We can then simplify our notation and write $\beta_k(\sigma) := \beta_k(e \to \sigma)$; $\mypath(\sigma):= \mypath(e\to\sigma)$.

Note that while permutations $\sigma$ and $d\sigma$, for $\sigma\in\SN$, $d\in D_N$, correspond to the same genome, $\beta_k(\sigma) \neq \beta_k(d\sigma)$ in general. A path probability $\beta_k(\sigma)$ gives the probability of a path of length $k$ between permutations, not genomes. 
For path probabilities between genomes, we need to consider all symmetries of the target genome. The reference frame is fixed by the reference genome, but we may have observed the target genome in any of its $2N$ distinct orientations with respect to this reference frame. Thus to find the path probability between genomes, we sum over the $2N$ symmetries of the target genome:
\[ \begin{split}
\alpha_k(\sigma) &:= \Prob(e\to [\sigma] \textrm{ via } k \textrm{ rearrangement events})\\
&= \sum_{d\in D_N} \beta_k(d\sigma)\,.	
\end{split} \]
Here we use `$e\to[\sigma]$' as shorthand for `$e\to(\textrm{any element of } [\sigma])$'. 

If all rearrangement events in $\M$ are equally probable, that is, $w(a) = \frac{1}{|\M|}$ for all $a\in\M$ then, in this case, for all $k\in\N$,
\[ \beta_k(\sigma) = \frac{|\mypath(\sigma)|}{|\M|^k}\,,\]
so that $\beta_k(\sigma)$ is simply counting the number of paths $a_k \ldots a_2 a_1 = \sigma$ for $\sigma\in\SN$.
In earlier work (for example \cite{jezandpet, mles}), this assumption of equal probability over the set $\M$ has been applied and, accordingly,  $\alpha_k(\sigma)$ and $\beta_k(\sigma)$ have been termed ``path counts'' rather than path probabilities.

The likelihood of our reference genome evolving into the genome $\sigma$ in a given time $T$ under the full biological model is now easily expressed as
\begin{equation}\label{eq gen like} \begin{split}
L(T|\sigma) &= \sum_{k=0}^{\infty} \Prob(e\to [\sigma] \textrm{ via } k \textrm{ events})\,\Prob(k \textrm{ events occur in time } T)\\
			&=\sum_{k=0}^{\infty} \alpha_k(\sigma)\,\Prob(k \textrm{ events occur in time } T)\,,
\end{split}
\end{equation}
where the remaining probablility is determined by the distribution $\dist$. 
For a given genome, this likelihood is a function of $T$ which we optimise to get the maximum likelihood estimate of elapsed time. 

We now recall the application of group representation theory, as introduced in \cite{jezandpet}, to convert this problem from a combinatorial to a numerical formulation. 
For background on the relevant aspects of group representation theory, we refer the reader to \cite{sagan}.

We denote by $\s$ the formal sum of the elements of $\M$, weighted by their probabilities, in the group algebra $\C[\SN]$. 
That is, 
\[ \s:=\sum_{a\in\M} w(a) a\,.\]
As discussed in \cite{jezandpet}, a pivotal role in what follows is played by the eigenvalues of $\s$.

We use $\rho$ and $\chi$ to denote representations and characters of $\SN$ respectively, along with their respective extensions to $\C[\SN]$. We denote the regular representation of $\SN$ (as carried by $\C[\SN]$) and its character by $\rho_{\reg}$ and $\chi_{\reg}$ respectively; we index, as usual, the irreducible representations and characters by integer partitions. 
Explicitly, $\rho_p$ refers to the irreducible representation of $\SN$ corresponding to the integer partition $p$ of $N$, where 
$p=(p_1^{m_1}, p_2^{m_2}, \ldots , p_k^{m_k})$, $p_1>p_2>\ldots>p_k>0$ and $\sum_{i=1}^k m_k p_k = N$. 
The set of partitions of $N$ will be denoted by $\Pa$.

We calculate path probabilities for a given $\sigma\in\SN$ by projecting onto the eigenspaces of the irreducible representations of $\s$. 
The derivation of this in \cite{jezandpet} was based on the assumption that each rearrangement in the set $\M$ occurs with equal probability. 
However, as foreshadowed in that paper, a short argument verifies that with our more general definitions of $\s$ and the path probabilities $\beta_k(\sigma)$, one still obtains the identity \cite[Section 3]{jezandpet}
\[ \s^k = \sum_{\sigma\in\SN} \beta_k(\sigma) \sigma, \]
and, accordingly,
\[ \beta_k(\sigma) = \tfrac{1}{N!} \chi_{\reg}(\sigma^{-1}\s^k) = \tfrac{1}{N!} \tr\left(\rho_{\reg}(\sigma^{-1})\rho_{\reg}(\s)^k\right)\,,\]
where $\tr$ denotes matrix trace. Thus
\begin{equation}\label{eq alpha irreps}
\begin{split}
\alpha_k(\sigma) = \tfrac{1}{N!}\sum_{d\in D_N}\beta_k(d\sigma) &= \tfrac{1}{N!}\sum_{d\in D_N} \,\sum_{p\in\Pa} D_p \chi_p(\sigma^{-1} d\, \s^k) \,,
\end{split}
\end{equation}
where $D_p$ is the dimension of the representation $\rho_p$. Then, as long as the matrices $\rho_p(\s)$ are diagonalisable (this is true, for example, if the model consists entirely of inversions or, more generally, is time reversible --- see Section~\ref{sec implement}), one has
\begin{equation}\label{eq alpha sum 1} 
\alpha_k(\sigma) = \tfrac{1}{N!}\sum_{d\in D_N}\, \sum_{p\in\Pa} D_p \,\sum_{\lambda_{p,i}} (\lambda_{p,i})^k \tr(\rho_p(\sigma^{-1} d) E_{p,i})\,.
\end{equation}
Here $\lambda_{p,i}$ is the $i$th eigenvalue of the irreducible representation $\rho_p(\s)$, and $E_{p,i}$ is the projection onto the eigenspace of $\lambda_{p,i}$.

%%% *** Note for me: in case the weights are constant, then the ``weighted'' $\s$ becomes just $\tfrac{1}{|\M|}$ times the path count $\s$. Then for each representation, each eigenvalue is accordingly weighted by $\tfrac{1}{|\M|}$, and **the projections remain unchanged**. 

\section{Model symmetries and genome equivalences}\label{sec symmetry}

Just as a genome $\sigma\in\SN$ has associated symmetries, $\{d\sigma\, :\, d\in D_N\}$, the biological model $(\M,w)$ may exhibit (or may be constructed to exhibit) certain symmetries. For example, if the model consists entirely of inversions, $a \!=\! a^{-1}$, then we have $\M^{-1}\!=\! \M$ (where $\M^{-1}\!:=\! \{a^{-1} \,:\, a\in\M\}$). 
Similarly, suppose that $a=(1,2)\in\M$. 
This is the rearrangement that swaps the regions in the first and second positions on the genome. 
For genomes with no distinguished region or position, we could reasonably expect that $\M$ also contains the rotated and flipped versions of this rearrangement, $(2,3), (3,4)$ and so on, that is, all of the rearrangements that swap regions in adjacent positions.  
For $\sigma\in\SN, a\in\M, d\in D_N$, we have
\[  a\sigma \equiv d (a \sigma) = d a (d^{-1} d) \sigma = (d a d^{-1}) (d \sigma) \,, \]
so the dihedral symmetries of a rearrangement $a\in\M$ are given by the rearrangements $\{dad^{-1}\,:\, d\in D_N\}$. 
%[perhaps some reference to how this is standard group action stuff and where to read about it]

\begin{definition}
Let $(\M, w)$ be a biological model on $N$ regions. 
\begin{enumerate}
\item[(i)] $(\M,w)$ is said to have {\em dihedral symmetry} if for all $d\in D_N$, 
\[ d\M d^{-1} :=\{dad^{-1}:a\in \mathcal{M}\}= \M \] 
and $w(dad^{-1}) = w(a)$ for all $a\in\M$.
\item[(ii)] $(\M,w)$ is said to be {\em time reversible} if $\M = \M^{-1}$ and $w(a^{-1}) = w(a)$ for all $a\in\M$.
\end{enumerate}
\end{definition}

We presently show that time reversibility of a model in conjunction with dihedral symmetry implies that the likelihood functions of $\sigma$ and $\sigma^{-1}\in\SN$ are identical. In this case, our distance measure (MLE) will be symmetric; that is,
for genomes $\sigma_1$ and $\sigma_2$, the time estimated for $\sigma_1$ to evolve into $\sigma_2$ is the same as that for $\sigma_2$ to evolve into $\sigma_1$. 

As mentioned in the previous section, it is common to consider only inversions as rearrangement events; in this case the model will be time reversible. It is almost trivial to construct examples of models that are time reversible but do not have dihedral symmetry;\footnote{If we do not require that $\M$ generates $\SN$, then it is certainly trivial!} for example, set $\M = \{(1,2), (2,3), \ldots , (N-1,N)\}\in\SN$.  
The following is an example of a model that has dihedral symmetry but is not time reversible. %This doesn't seem like a particularly plausible model biologically, but hey.

\begin{example}
Let $N\!=\!5$ and consider $a\!=\!(1,2,4,3), b=(1,2)\in\SN$. 
Define
\[ \begin{split}
\M &\!:= \{d a d^{-1}, d b d^{-1}\,:\, d\in D_N\} \\
	&\,= \{(1,2,4,3),(2,3,5,4),(3,4,1,5),(4,5,2,1),(5,1,3,2), (1,2), (2,3),(3,4),(4,5), (5,1)\}\,, \end{split}
\] 
and note that the rotations of $a$ and $b$ correspond to their reflections. Clearly $d\M d^{-1} = \M$ for all $d\in D_N$, so defining $w$ to be constant on $\M$ gives a model with dihedral symmetry. 
However, $a^{-1} = (1,3,4,2)\notin\M$, so the model is not time reversible. 

This example easily generalizes to any $N\geq 5$: $a=(1,2,4,3), b=(1,2)\in\SN$, but $a^{-1}\notin \M := \{d a d^{-1}, d b d^{-1}\,:\, d\in D_N\}$. Note that if we do not require $\M$ to generate $\SN$, we do not need to include the element $b$ and its symmetries in the model set.    
\hfill $\diamond$
\end{example}

As the permutations $\sigma$ and $d\sigma\in\SN$ represent the same genome for any $d\in D_N$, it is intuitively clear that we must always have $\alpha_k(\sigma) =\alpha_k(d\sigma)$ (to see this algebraically, simply note that, as $D_N$ is a group, $[\sigma]=[d\sigma]$). 
If the set $\M$ has some further symmetry, and the probabilities $w(a)$ reflect this symmetry, then there will be many more genomes (beyond those in $[\sigma]$) that have the same path probabilities, and hence likelihood functions, as $\sigma$. 
The following result will allow us to split $\SN$ into further equivalence classes of permutations corresponding to genomes with the same likelihood functions.

\begin{proposition}\label{prop path probs equality}
Let $(\M, w)$ be a biological model on $N$ regions with dihedral symmetry. Let $\sigma\in\SN$ and $k\in\N$.
\begin{enumerate}
\item[(i)] For any $d_1, d_2\in D_N$, $\alpha_k(d_1\sigma d_2) = \alpha_k(\sigma)$.
\item[(ii)] If, further, the model is time reversible, then $\alpha_k(\sigma^{-1}) = \alpha_k(\sigma)$.
\end{enumerate}	
\end{proposition}
\begin{proof}
\begin{enumerate}[wide, labelwidth=!, labelindent=0pt]
\item[(i)] Let $d\in D_N$. If $\sigma \!=\! a_k\ldots a_2 a_1$, where each $a_i\in\M$, then we have $d\sigma d^{-1} \!=\! (da_k d^{-1})\ldots(d a_{2} d^{-1})(d a_1 d^{-1})$. By dihedral symmetry, each $d a_i d^{-1} \in\M$. 
Conversely, if $d\sigma d^{-1} \!=\! a_k\ldots a_2 a_1$, where each $a_i\in\M$, then
\[ \sigma = d^{-1}a_k\ldots a_2 a_1 d =   (d^{-1}a_k d) \ldots (d^{-1}a_2 d) (d^{-1} a_1 d) \,\]
and each $d^{-1}a_i d = d^{-1}a_i (d^{-1})^{-1} \in\M$. 
This verifies that, for each path $\varphi\in\mypath(\sigma)$, there exists a path $\varphi' \in\mypath(d\sigma d^{-1})$, and vice versa. Dihedral symmetry ensures that these paths have equal probability; for any $d\in D_N$, we have
\[ \begin{split}
\Prob((da_1 d^{-1} ,d a_{2} d^{-1},\ldots,d a_k d^{-1} ))
&= w(da_1 d^{-1}) w(d a_{2} d^{-1})\ldots w(d a_k d^{-1})\\
&= w(a_1) w(a_2)\ldots w(a_k) \\
&= \Prob((a_1 ,a_2,\ldots, a_k))\,.
\end{split} \]
Thus
\begin{equation}\label{eq beta conj}
	\beta_k(d\sigma d^{-1}) =\beta_k(\sigma) \;\; \textrm{ for all } \; d\in D_N \,. 
\end{equation}
Now, let $d_1, d_2\in D_N$. We have
\[
\alpha_k(d_1\sigma d_2)=\sum_{d\in D_N} \beta_k(d (d_1\sigma d_2)) =\sum_{d\in D_N} \beta_k(d_2 (d d_1\sigma d_2)d_2^{-1})=\sum_{d\in D_N} \beta_k(d \sigma )= \alpha_k(\sigma),
\]
where, along with (\ref{eq beta conj}), we have used that $D_N$ is a group, so $\{d_2 d d_1 : d\in D_N\} \!=\! D_N$.
\smallskip
\item[(ii)] 
If the model is time reversible, then $\sigma \!=\! a_k\ldots a_2 a_1$, with each $a_i\in\M$, if and only if $\sigma^{-1}= a_1^{-1} a_2^{-1} \ldots a_k^{-1}$, with each $a_i^{-1}\in\M$. As $w(a_i) = w(a_i^{-1})$ for all $i$, it follows that
\begin{equation}\label{eq beta inv}
	\beta_k(\sigma)=\beta_k(\sigma^{-1})\,.
\end{equation}
Then
\[
\alpha_k(\sigma^{-1}) =\sum_{d\in D_N} \beta_k(d \sigma^{-1})=\sum_{d\in D_N} \beta_k(\sigma d^{-1})=\sum_{d\in D_N} \beta_k(\sigma d)=\sum_{d\in D_N} \beta_k(d\sigma)=\alpha_k(\sigma),
\]
where we have used (\ref{eq beta inv}), $D_N^{-1} = D_N$, and (\ref{eq beta conj}). 
\end{enumerate}
\end{proof} 

The following example shows that dihedral symmetry, along with time reversibility, is strictly necessary to obtain $\alpha_k(\sigma) = \alpha_k(\sigma^{-1})$.

\begin{example}
Let $N=6$, $\M = \{(1,2),(2,3),(3,4),(4,5),(5,6)\}$ and $w(a) = \tfrac{1}{5}$ for all $a\in\M$. This model is time reversible without dihedral symmetry (since $(6,1)\notin \mathcal{M}$). Consider $\sigma=(1,2,4,3)$ and $\sigma^{-1}= (1,3,4,2)\in\SN$. Using SageMath and (\ref{eq alpha sum 1}) to calculate the path probabilities, we found that $\alpha_4(\sigma) = \tfrac{11}{5^4}$ and $\alpha_4(\sigma^{-1}) = \tfrac{8}{5^4}$.
\hfill $\diamond$
\end{example}

Examining the sum (\ref{eq gen like}), we observe that the likelihood functions of $\sigma_1$ and $\sigma_2\in\SN$ coincide if and only if all path probabilities coincide, that is, if and only if $\alpha_k(\sigma_1)=\alpha_k(\sigma_2)$ for all $k$.

\begin{definition}
	Let $(\M, w)$ be a biological model on $N$ regions. Two permutations $\sigma_1, \sigma_2\in\SN$ are said to be {\em $(\M, w)$-equivalent} if their likelihood functions calculated under the model are equal.
\end{definition}

Now, in view of the observation above, we apply Proposition~\ref{prop path probs equality} and the preceding discussion to gain the following.

\begin{corollary}\label{cor equiv classes}
	Let $(\M, w)$ be a biological model on $N$ regions. The $(\M, w)$-equivalence classes of $\SN$ are as follows.
	\begin{enumerate}
		\item[(i)] For every choice of model, 
		\[ [\sigma] = \{d\sigma\,:\, d\in D_N\}\,,\;\; \sigma\in\SN\,.\]
		\item[(ii)] For models with dihedral symmetry,
		\[ [\sigma]_D := \{d_1\sigma d_2\,:\, d_1,d_2\in D_N\}\,,\;\; \sigma\in\SN\,.\]
		\item[(iii)] For models with dihedral symmetry that are time reversible,
		\[ [\sigma]_{DR}:= \{d_1\sigma d_2\,:\, d_1,d_2\in D_N\}\cup\{d_1\sigma^{-1} d_2\,:\, d_1,d_2\in D_N\}\,,\;\; \sigma\in\SN\,.\] $\hfill\square$
	\end{enumerate}
\end{corollary}

Note that the equivalence classes are independent of the chosen distribution of events in time. Further, as they may equivalently be defined in terms of path probabilities, they are valid for other measures of genomic distance that depend on path probabilities (or path counts), such as mimimum distance.

If the biological model exhibits no symmetry, then each equivalence class will simply correspond to a single genome (case (i) above). In cases (ii) and (iii) above (for $N>3$), an equivalence class will, in general, contain permutations corresponding to different genomes. Note that although these genomes will have the same MLE, that is, be an identical distance from the reference genome, their distance to one another as calculated under the model will in general not be zero. (Algebraically, we are simply noting that, in general, for $\sigma_1, \sigma_2\in [\sigma]_{x}$, $\sigma_2\sigma_1^{-1}\not\in[e]_x$, where $x=D$ or $DR$.)

For models with dihedral symmetry, Corollary~\ref{cor equiv classes} (ii) generalises \cite[Proposition 4]{mles}, which proves (\ref{eq beta conj}) for any element $d$ in the \emph{normaliser} of $\M$, under the assumption of constant probability on $\M$ (i.e. $w(a)=\text{const.}$).\footnote{The result in \cite{mles} is stated in terms of likelihoods, however the likelihoods are for single elements of $\SN$, with dihedral symmetry not included in calculations until later in the paper.}
Recall that the normaliser of a set $\M\in\SN$ is the group
\[\No_{\SN}(\M):=  \{g\in\SN\,:\, g\M g^{-1} = \M\}\,.\]
For a model $(\M,w)$ with dihedral symmetry, $D_N\subseteq\No_{\SN}(\M)$. 
In fact, our proof of (\ref{eq beta conj}) easily extends to all elements in the normaliser of $\M$, provided that for $a\in\M$, $g\in \No_{\SN}(\M)$, one has $w(g a g^{-1}) = w(a)$.
It may be the case that the normaliser of $\M$ is strictly larger than $D_N$. However, in practice we are not so interested in other (non-dihedral) elements in the normaliser of $\M$, as result (i) of Proposition~\ref{prop path probs equality} is not generalisable to the whole normaliser.  

We further note that a version of \cite[Proposition 4]{mles} is included in \cite[Section 3]{jezandpet}. The statements in these two papers are connected by the fact that, in the case of constant probability on $\M$, the stabiliser of the group algebra element $\s$ is exactly the normaliser of the set $\M$.
That is, 
\[ \{g\in\SN\,:\, g \s g^{-1} = \s\} = \{ g\in\SN\,:\, g\M g^{-1} = \M\} \,.\]

\subsection{Counting the equivalence classes}

As observed above, if we do not assume any symmetry in our model, then the number of equivalence classes of permutations under the model is simply the number of distinct genomes. 

\begin{proposition}\label{prop equiv genomes}
For any biological model on $N\geq 3$ regions, the number of genome equivalence classes $[\sigma]\subseteq\SN$ is $\tfrac{N!}{2N}=\tfrac{(N-1)!}{2}$. 
$\hfill\square$
\end{proposition}

The sequence $\left(\tfrac{(N-1)!}{2}\right)$ also gives the order of the alternating group $A_{N-1}$. In the Online Encyclopaedia of Integer Sequences (OEIS) \cite{oeis}, this is sequence A001710.\footnote{We note that the OEIS entry includes a characterisation of this sequence that is equivalent to our definition of genomes (namely, the number of necklaces that may be formed from $N$ distinct beads).}
Finding the number of equivalence classes when we have some model symmetry is somewhat less trivial.

For a model with dihedral symmetry, using SageMath and Corollary~\ref{cor equiv classes} we found that the number of equivalence classes $[\sigma]_D \subseteq\SN$ for $N=3,4,5,\ldots 10$ is
\[ 1, 2, 4, 12, 39, 202, 1219, 9468, \ldots\,.\]
This coincides with the beginning of sequence A000940 from the OEIS \cite{oeis}, which in \cite{polygons} is described as the number, $S(N)$, of classes of similar polygons on $N$ vertices. The paper further gives exact expressions for $S(N)$ for odd and even values of $N$. We now proceed to show that the number of equivalence classes $[\sigma]_D \subseteq\SN$ is indeed given by $S(N)$.
The relevant definitions from \cite{polygons} are as follows. 

\begin{definition}
Given $N$ equally spaced points on a circle, one forms a {\em polygonal path} by choosing a labelling of the points with $1,2,\ldots,N$ and forming the (directed) edges $N\rightarrow 1$ and $i \rightarrow i+1$ for $i=1,2,\ldots , N-1$. 
Ignoring the labels of the points and the direction of the edges, one has a {\em polygon}. Thus, two polygonal paths that differ only in starting point or orientation of numbering are said to define {\em identical polygons} and two polygons that differ only by a plane rotation or a reflection through an axis are termed {\em similar polygons}.
\end{definition}
Note that identical polygons coincide with an unlabelled shape with fixed orientation, whereas similar polygons coincide with an unlabelled shape without fixed orientation.
The key observation in the following is that the dihedral group $D_N$ acts on polygonal paths by either permutating the labels (acting on the right) or by physical rotation and/or reflection (acting on the left).

\begin{proposition}\label{prop equiv D}
For a model with dihedral symmetry, the number of equivalence classes $[\sigma]_D \subseteq\SN$ is exactly the number, $S(N)$, of classes of similar polygons on $N$ vertices.
\end{proposition}
\begin{proof}
A permutation $\sigma\in\SN$ may be represented as $N$ equally spaced, labelled points on a circle by choosing a reference frame and placing region label $i$ on the point $\sigma(i)$. 
Forming, as above, an edge $i\to i+1$ for $i=1,2,\ldots N-1$ and an edge $N\to 1$ defines a polygonal path and thus a polygon $q_{\sigma}$.
The set of polygons  identical to $q_{\sigma}$ is defined in exactly this way by the permutations $\{\sigma d \,:\, d\in D_N\}$. The set of plane rotations and reflections of these is given by the permutations $\{d_1 \sigma d\,:\, d_1, d\in D_N\}$.

Conversely, given a polygon $q$ on $N$ equally spaced vertices on a circle, choose a starting point and orientation of numbering and label the vertices $1,\ldots , N$. This defines a permutation $\sigma_q \in \SN$ where $\sigma_q(i)\! =\! j$ whenever vertex label $i$ is in position $j$. 
Reasoning as above, we see that the set of polygons similar to $q$ defines the equivalence class of permutations $[\sigma_q]_D$.
\end{proof}

%\subsection{Counting the dihedral, time reversible equivalence classes $[\sigma]_{DR}$}

For a model with dihedral symmetry and time reversibility, using SageMath and Corollary~\ref{cor equiv classes} we found that the number of equivalence classes $[\sigma]_{DR} \subseteq\SN$ for $N=3,4,5,\ldots , 10$ is
\[ 1, 2, 4, 10, 28, 127, 686, 4975,\ldots\,.\]
This coincides with sequence A006841 from the OEIS \cite{oeis}, the number, $T(N)$, of inequivalent permutation arrays of size $N$. An expression for $T(N)$ is derived in \cite{arrays}. We now proceed to show that the number of equivalence classes $[\sigma]_{DR} \subseteq\SN$ is indeed given by $T(N)$. We begin with the relevant definitions from \cite{arrays}.

\begin{definition}
A {\em permutation array of period} $N$ is an $N\times N$ matrix with a single `$1$' in each row and column and `$0$'s elsewhere, that is, a permutation matrix. Two permutation matrices are said to be {\em equivalent} if one can be obtained from the other by a cyclic shift of rows or columns, by rotation, by transposition, or by any  sequence of these operations.
\end{definition}

We note that the number of equivalence classes $T(N)$ is always greater than $\frac{N!}{8N^2}$ and approaches this value asymptotically \cite{oeis}. 

\begin{proposition}\label{prop equiv DT}
For a model with dihedral symmetry that is time reversible, the number of equivalence classes $[\sigma]_{DT} \subseteq\SN$ is exactly the number, $T(N)$, of classes of equivalent permutation arrays of dimension $N$.
\end{proposition}
\begin{proof}
Let $\sigma\in\SN$ and $A_{\sigma}$ be the corresponding permutation matrix, that is, $(A_{\sigma})_{ij} \!=\! 1$ if $\sigma(i) \!=\! j$ and $(A_{\sigma})_{ij} \!=\! 0$ otherwise. Firstly observe that $A_{\sigma^{-1}} = (A_{\sigma})^T$ and that if $B$ is the result of rotating a matrix $A$ clockwise by ninety degrees, we have
\[ (B)_{ij} = (A)_{N+1-j,i}\,,\]
so that this rotation is equivalent to transposition of $A$ followed by vertical reflection, and (or) to horizontal reflection followed by transposition.

Define the usual generators $r,s\in D_N$ by
\[ r = (1,2,\ldots N) \,, \quad s = (1,N)(2,N-1)\ldots \left\{ \begin{array}{l} (\tfrac{N}{2},\tfrac{N}{2}+1),\; N \textrm{ even}\\[.4ex]
																				(\tfrac{N-1}{2},\tfrac{N+3}{2}),\; N \textrm{ odd}. \end{array} \right. \]
Then $D_N = \{r^is^j\,:\, i=0,1,\ldots N-1,\, j=0,1 \}$. 

Recalling that for $a,b\in\SN$, $A_{ab} = A_a A_b = $ (the matrix $A_a$ with columns permuted according to $b) = $ (the matrix $A_b$ with rows permuted according to $a$), we see that 
\[ \begin{split}
A_{\sigma r} &= \left(A_{\sigma} \textrm{ with columns shifted cyclically to the left}\,\right); \\
A_{r\sigma} &= \left(A_{\sigma} \textrm{ with rows shifted cyclically up}\,\right); \\
A_{\sigma s} &= \left(A_{\sigma} \textrm{ reflected vertically} = A_{\sigma} \textrm{ transposed}  \textrm{ then rotated by } 90^{\circ}\,\right); \\
A_{s \sigma} &= \left(A_{\sigma} \textrm{ reflected horizontally} =  A_{\sigma} \textrm{ rotated by } 90^{\circ} \textrm{ then transposed} \,\right).
\end{split} \]
Then for any $d_1,d_2\in D_N$, the matrix $A_{d_1\sigma d_2}=A_{r^i s^j \sigma r^k s^{\ell}}$ is the matrix gained from $A_\sigma$ by a sequence of column shifts, rotations, transpositions, and row shifts. Thus the permutation arrays defined by $d_1\sigma d_2$, $d_1\sigma^{-1} d_2$ and $\sigma$ are equivalent.

Conversely, any given permutation array is a permutation matrix $A_{\sigma}$, for some $\sigma\in\SN$, and any matrix equivalent to $A_{\sigma}$ is another permutation matrix $A_{\sigma_m}$, gained from $A_{\sigma}$ by a sequence of moves defining a sequence of permutations as follows: $\sigma_1\!=\!\sigma$  and, for $i=2,3,\ldots m$, $\sigma_i$ is equal to one of $r\sigma_{i-1}$, $s\sigma_{i-1}$, $\sigma_{i-1}r$, $\sigma_{i-1}s$ or $(\sigma_{i-1})^{-1}$. Thus, as $D_N$ is a group, $\sigma_m = d_1 \sigma^{\pm 1} d_2$ for some $d_1 , d_2 \in D_N$, and thus $\sigma_m \in [\sigma]_{DR}$ .
\end{proof}

In practice, the models we consider as biologically reasonable are time reversible with dihedral symmetry.
Thus Corollary~\ref{cor equiv classes} reduces the number of genomes for which we must calculate likelihoods from $\frac{N!}{2N}$ to the order of $\frac{N!}{8 N^2}$.

Propositions~\ref{prop equiv genomes},~\ref{prop equiv D} and \ref{prop equiv DT} provide a link between the corresponding OEIS sequences in terms of genome rearrangement model symmetries. 
This can also be described in terms of classes of permutation matrices as follows.

\begin{definition}

\begin{enumerate}
\item[ ]
\item[(i)] Two permutation matrices are said to be {\em genome-equivalent} if one may be obtained from the other by horizontal reflection, a cyclic shift of rows, or by any  sequence of these operations.
\item[(ii)] Two permutation matrices are said to be {\em D-equivalent} if one may be obtained from the other by vertical or horizontal reflection, a cyclic shift of rows or columns, or by any  sequence of these operations. 
\item[(iii)] Two permutation matrices are said to be {\em DR-equivalent} if one may be obtained from the other by vertical or horizontal reflection, a cyclic shift of rows or columns, transposition, or by any  sequence of these operations. 
\end{enumerate} 
\end{definition}

The next result follows directly from Propositions~\ref{prop equiv genomes} and \ref{prop equiv D} and the proof of Proposition~\ref{prop equiv DT}.

\begin{corollary}
The numbers of classes of (i) genome-equivalent, (ii) D-equivalent, and  (iii) DR-equivalent permutation matrices are given by the numbers of equivalence classes of the form (i) $[\sigma]$, (ii) $[\sigma]_{D}$, and (iii) $[\sigma]_{DR}$, respectively. 
$\hfill\square$
\end{corollary}

\section{Exploiting symmetry in the dihedral sum}\label{sec dihedral}

We now return to the calculation of the path probabilities $\alpha_k(\sigma)$ defined in Section~\ref{sec technique}.
Since one must sum over the dihedral symmetries of a genome,  here we exploit the fact that we are working with {\em linear} representations of $\SN$ and find that, rather than summing over the dihedral group at the final stage of the calculation, it is more efficient to do it earlier. 

\begin{example}\label{eg S5}
Let $\M:=\{(1,2),(2,3),(3,4),(4,5),(5,1)\}\subseteq\mathcal{S}_5$ and define $w(a) = \tfrac{1}{5}$ for each $a\in\M$. 
Then $(\M,w)$ is a time reversible model with dihedral symmetry.
While $\mathcal{S}_5$ has 120 elements, meaning the regular representation of $\mathcal{S}_5$ has dimension 120, we may decompose this into seven irreducible representations with dimensions $1, 4, 5, 6, 5, 4$ and $1$.
In this case, one may obtain exact eigenvalues (in algebraic form) for each irreducible representation of $\s$. 
Thus, extracting the relevant parts of (\ref{eq alpha sum 1}), 
\[ \beta_k(\sigma) =\tfrac{1}{N!}\sum_{p\in\Pa} D_p \,\sum_{\lambda_{p,i}} (\lambda_{p,i})^k \tr(\rho_p(\sigma^{-1}) E_{p,i})\,,\]
we can obtain exact expressions for path probabilities. For example, for $\sigma_1 \!=\! e$, $\sigma_2 \!=\! (1,2,3,4,5)$, and  $\sigma_3 \!=\! (1,3,5,2,4)$, and for even values of $k$, we have
\[ 
\begin{array}{l}
\begin{split}
	\beta_k(\sigma_1)= \tfrac{1}{60. 5^k} (5^k + 5 + 6 \sqrt{5}^{k} + 10 ((1 +\sqrt{5})^{k} + & (1-\sqrt{5})^{k}) \\
	&+ \tfrac{8}{2^k}((5+\sqrt{5})^{k} + (5-\sqrt{5})^{k}) )\,;
	\end{split}\\ 
 \begin{split}
	\beta_k(\sigma_2) = \tfrac{1}{120. 5^k} (2.5^{k} +10+12\sqrt{5}^k -& 5((1+\sqrt{5})^{k+1} + (1-\sqrt{5})^{k+1}) \\
	&+ \tfrac{16\sqrt{5}}{2^k}((5+\sqrt{5})^{k-1} -(5-\sqrt{5})^{k-1}))\,; \end{split} \\
 \begin{split} 
	\beta_k(\sigma_3) = &\tfrac{1}{120. 5^k} (2.5^{k} + 10 +12\sqrt{5}^k + 20((1+\sqrt{5})^{k-1} + (1-\sqrt{5})^{k-1}) \\
	&- \tfrac{4\sqrt{5}}{2^k}((5+\sqrt{5})^{k} -(5-\sqrt{5})^{k})- \tfrac{4}{2^k}((5+\sqrt{5})^{k} + (5-\sqrt{5})^{k})) \,.
	\end{split}
	\end{array}
	\]
For odd values of $k$, $\beta_k(\sigma_i)$ is zero in each of these  cases.

Now these three permutations all represent the same genome, as $(1,2,3,4,5)$ and $(1,3,5,2,4)$ are both rotations of the identity, $e$. Summing these three expressions together with those for the other seven symmetries of $e$ gives the much simpler expression
\begin{equation}\label{eq S5 pp(e)}
	\alpha_k(e) = \left\{ \begin{array}{cl} \tfrac{1}{6. 5^k}(5^k + 5) \,, & k \textrm{ even}; \\
	0 \,,				&  k \textrm{ odd}\,. \end{array} \right.  
\end{equation}
\hfill $\diamond$
\end{example}

To see that we may in general bypass the calculation of the $\beta_k(\sigma)$, we denote by $\D$ the formal sum of the elements of the dihedral group in the group algebra $\C[\SN]$, that is, $\D:=\sum_{d\in D_N} d$. 
Then from equation (\ref{eq alpha sum 1}),
\begin{align}
\alpha_k(\sigma)	%= \sum_{d\in D_N}\, \beta_k(d\sigma)
					&= \tfrac{1}{N!}\sum_{d\in D_N}\, \sum_{p\in\Pa} D_p \,\sum_{\lambda_{p,i}}\, (\lambda_{p,i})^k \tr(\rho_p(\sigma^{-1} d) E_{p,i})\nonumber\\ 
					&= \tfrac{1}{N!}\sum_{p\in\Pa} D_p \,\sum_{\lambda_{p,i}}\, (\lambda_{p,i})^k \,\sum_{d\in D_N}\tr(\rho_p(\sigma^{-1} d) E_{p,i})\nonumber\\ 
					&= \tfrac{1}{N!}\sum_{p\in\Pa} D_p \,\sum_{\lambda_{p,i}}\, (\lambda_{p,i})^k \,\tr(\rho_p(\sigma^{-1} \D) E_{p,i})\,,\label{eq alpha sum 2}
\end{align}
where we have simply swapped the order of the finite sums and applied linearity of the trace.

As well as generally reducing the amount of computation (and thus reducing error in numerical computations), the form of (\ref{eq alpha sum 2}) allows us to identify certain partitions $p$ such that  $\rho_p(\D)\!=\!0$, which, in turn, implies $\tr(\rho_p(\sigma^{-1}\D))=\tr(\rho_p(\sigma^{-1})\rho_p(\D))=0$ for all choices of $\sigma$. 
We show that this can be achieved via a quite straightforward calculation applying Frobenius' character formula (see for example \cite[Chapter 4]{fulhar}), as follows.

The key observation is that the dihedral group acts trivially on the element $\mathbf{d}\in \C[\SN]$.
That is, for each $d'\in D_N$ we have
\[
d'\mathbf{d}=\sum_{d\in D_N}d'd=\sum_{d\in D_N}d=\mathbf{d}.
\]
This observation of course translates to the irreducible representations, so that $\rho_p(d'\mathbf{d})=\rho_p(\mathbf{d})$ for all partitions $p$.
Thus if we consider, for each irreducible representation $p$ of $\SN$, the restriction of $\rho_p$ to the dihedral group $D_N$,\footnote{simply defined via the matrices $\rho_p(d)$ for $d\in D_N$} we see that: 
\begin{equation}
\label{eq:deq0}
\left(\rho_p \text{ carries no copy of the trivial representation of }D_N\right)\Longrightarrow\rho_p(\mathbf{d})= 0.
\end{equation}

We recall that the irreducible representations of any finite group $G$ have orthogonal characters under the inner product
\[
\langle \chi,\chi' \rangle:=\tfrac{1}{|G|}\sum_{g \in G}\overline{\chi(g)}\chi'(g).
\]
For each partition $p$, the irreducible $\SN$-character $\chi_p$ provides, under restriction to $D_N$, a character for $D_N$.
Although this character need no longer be irreducible in general, the trivial $\SN$-character $\chi_{(N)}\equiv 1$, corresponding to the partition $p\!=\!(N)$, certainly does remain irreducible.
Thus, under the inner product for $D_N$-characters,
\[
\langle \chi_p,\chi_{p'} \rangle_{D_N}:=\tfrac{1}{2N}\sum_{d \in D_N}\overline{\chi_p(d)}\chi_{p'}(d),
\]
we have, applying (\ref{eq:deq0}),
\begin{equation}\label{eq:whydis0}
\langle \chi_{(N)},\chi_{p} \rangle_{D_N}=0 \implies \rho_p(\D)=0.
\end{equation}

We write the conjugacy classes of $\SN$ as $[c_{\mathbf{i}}]$, where $c_{\mathbf{i}}$ is a permutation with cycle structure $\mathbf{i}=(i_1, \cdots , i_N)$; that is, $c_{\mathbf{i}}$ has $i_j$ cycles of length $j$ for each $j$. 

In general, we have
\begin{equation}\label{eq:dsum}
\langle \chi_{(N)},\chi_p \rangle_{D_N} = \sum_{[c_{\ii}]} \left|[c_{\ii}]\right|\chi_p(c_{\ii})
\end{equation}
and by duality,\footnote{$\rho_{p^{*}}(\sigma):=\sgn(\sigma)\rho_p(\sigma)$.} 
\begin{equation}\label{eq:dsum dual}
\langle \chi_{(N)},\chi_{p^*} \rangle_{D_N} = \sum_{[c_{\ii}]} \left|[c_{\ii}]\right|\sgn(c_{\ii})\chi_p(c_{\ii})\,,
\end{equation}
where we write $\sgn\equiv \chi_{(1^N)}$ for the character of the sign representation of $\SN$.

Now via Frobenius' character formula and some tedious algebra, we derived expressions for $\chi_p(c_{\ii})$ in terms of components of $\ii$, for a selection of partitions $p$ (based on empirical observations from calculations of $\rho_p(\D)$ in SageMath):
\begin{equation}\label{eq:char forms}
\begin{aligned}
\chi_{(N-1,1)}(c_{\ii}) 	&= i_1-1\,, \\
\chi_{(N-2,1,1)}(c_{\ii}) 	&= \tfrac{(i_1-1)(i_1-2)}{2} - i_2\,, \\
\chi_{(N-2,2)}(c_{\ii}) 	&= \tfrac{i_1(i_1-3)}{2} + i_2\,, \\
\chi_{(N-3,3)}(c_{\ii}) 	&= \tfrac{i_1(i_1-1)(i_1-5)}{6} + i_2(i_1-1) + i_3\,.
\end{aligned}
\end{equation}

Using the expressions in (\ref{eq:dsum}), (\ref{eq:dsum dual}) and (\ref{eq:char forms}), evaluating on the conjugacy classes of $D_N$, and applying (\ref{eq:whydis0}), we obtained the general results presented in Table~\ref{tab:frob}.

\begin{table}[t]
\begin{tabular}{l|l||l|l}
	\hline
	partition $p$ 	& $\rho_p(\D) = 0$ for & dual partition $p^*$ & $\rho_{p^*}(\D) = 0$ for \\
	\hline \hline
	$(N)$ 			& no $N$		& $(1^N)$			& $N\neq 4k+1$, $k\in \N$ \\
	\hline
	$(N-1,1)$		& all $N$ 		& $(2,1^{N-2})$		& $N\neq 4k+2$, $k\in\N$ \\
	\hline
	$(N-2,2)$		& no $N$ 		& $(2^2,1^{N-4})$	& $N=6$ or $N=4k+3$, some $k\in\N$ \\
	\hline
	$(N-2,1^2)$		& all $N$ 		& $(3,1^{N-3})$ 	& $N=4k+1$, some $k\in\N$ \\
	\hline
	$(N-3,3)$		& $N=6$  		& $(2^3,1^{N-6})$ & no $N$  
\end{tabular}\smallskip
\caption{Irreducible representations of $\SN$ with $\rho(\mathbf{d})=0$.}\label{tab:frob}
\end{table}

As an example, in $\mathcal{S}_5$, the irreducible representations are indexed by the partitions $(5), (4,1), (3,2), (3,1^2), (2^2,1), (2,1^3)$ and $(1^5)$.
Checking Table~\ref{tab:frob}, we see that $\rho_p(\D) = 0$ for $p= (4,1), (3,1^2)$ and $(2,1^3)$, so one need only calculate the sum (\ref{eq alpha sum 2}) over four partitions rather than seven. 
For such small values of $N$, this result can greatly reduce computation, however, for all cases up to $N\!=\!16$, we found that the number of irreducible representations $\rho_p$ such that $\rho_p(\D) \!=\! 0$ does not increase above four or five. 
Thus these results, although algebraically interesting, do not appear to significantly reduce computation for large values of $N$. 
However, they remain useful in providing some known partial trace values against which we may check our calculated numerical values.
Additionally, the general principle remains that computing with the matrices $\rho_p(\D)$ (as in (\ref{eq alpha sum 2})) is superior to summing over $D_N$ independently (as in (\ref{eq alpha sum 1})).

\section{Implementation and results}\label{sec implement}

In this section, we describe the practical process of implementing the above ideas, and present a selection of our results. 
All algebraic calculations were undertaken in the open source package SageMath \cite{sage}, installed on an instance of the Nectar Research Cloud running Ubuntu 16.04.2 with 64 gigabytes of available RAM. 
SageMath has the inbuilt capability to calculate the irreducible representations of symmetric group elements.\footnote{underlying code written by Franco Saliola} 
All plots and maximum likelihood estimates were produced using R~\cite{R} running on a standard desktop machine. 
Our code and the complete set of results are provided in the supplementary material.

Beyond the raw number of regions, $N$, the main drivers of computational complexity in the calculation of the likelihood function $L(T|\sigma)$ are the dimensions of the irreducible representations of $\mathcal{S}_N$.
Hence we note that, theoretically, the feasibility of the calculation is independent of the choice of model. 
We verify that this holds in practice by undertaking calculations for two distinct biological models $(\mathcal{M}_1,w_1)$ and $(\mathcal{M}_2,w_2)$, each of which is defined below.

To specify the full biological model, in each case we set $\dist \equiv \Pois(1)$; that is, we suppose that rearrangement events are distributed in time according to a Poisson distribution with the expected number of events per unit of time $T$ equal to 1. 
Following \cite{jezandpet} and applying (\ref{eq alpha sum 2}) to (\ref{eq gen like}), the likelihood function  now takes the form 
\begin{equation}\label{eq like} \begin{split}
L(T|\sigma) 	&= \sum_{k=0}^{\infty} \alpha_k(\sigma)\,\frac{T^k e^{-T}}{k!}\\
&= \frac{1}{N!}\sum_{k=0}^{\infty} \sum_{p\in\Pa} D_p \,\sum_{\lambda_{p,i}}\, (\lambda_{p,i})^k \,\tr(\rho_p(\sigma^{-1} \D) E_{p,i})\,\frac{T^k e^{-T}}{k!}\\
&= \frac{\;e^{-T}}{N!} \,\sum_{p\in\Pa} D_p \,\sum_{\lambda_{p,i}}\,\tr(\rho_p(\sigma^{-1} \D) E_{p,i}) e^{\lambda_{p,i} T}\,.
\end{split} \end{equation}
Note that, as was observed in \cite{jezandpet}, the infinite sum present in (\ref{eq gen like}) has been converted to a finite sum. 
This is in contrast to approach taken in \cite{mles} where the likelihood functions were approximated by computing $\alpha_k(\sigma)$ for finitely many values and truncating the infinite sum. 

Our first biological model $(\M_1,w_1)$  assumes each adjacent inversion is possible and equally likely:
\[ \M_1 = \{ (1,2), (2,3),  \ldots , (N\!\!-\!\!1,N),(N,1) \} \subseteq \SN\,, \]
together with $w_1(a)\!=\!\frac{1}{N}$ for each $a\in \M_1$.
This model has been considered previously in \cite{attilaand,mles,jezandpet}. 

Our second biological model allows, as well as inversions of adjacent regions, inversions of \emph{three} adjacent regions. That is, we set
\[ \mathcal{M}_2 \!=\! \{ (1,2), (2,3),  \ldots , (N\!-\!1,N),(N,1) \} \cup \{(1,3),(2,4),\ldots,(N\!-\!1,1),(N,2)\} \subseteq \SN\,. \]
Reflecting the empirical observation that inversions of larger regions are statistically less likely \cite{darling08}, we set 
\begin{equation}
\label{eq:w2}
w_2(a) = \left\{ \begin{array}{ll} 
						\tfrac{2}{3}\cdot\tfrac{1}{N}\,,& \quad a = (1,2), (2,3),  \ldots , (N,1)\,;\\[.4ex]
						\tfrac{1}{3}\cdot\tfrac{1}{N}\,,& \quad a = (1,3), (2,4),  \ldots , (N,2)\,. \end{array} \right. 
\end{equation}						
Note that we make these choices for illustrative purposes only and are not claiming that these probabilities are biologically realistic; we simply need to specify {\em some} probabilities in order to make calculations under the models. 

We observe that both models $(\M_1,w_1)$ and $(\M_2,w_2)$ have dihedral symmetry and are time reversible. 

We also need to verify that the representation matrices $\rho_p(\s)$ are  diagonalisable.
This is true as they are symmetric: each of $\M_1$ and $\M_2$ consists entirely of inversions, $a = a^{-1}$, and thus we have
\[ \rho_p(\s)^T = \sum_{a\in\M}\!\!w(a)\rho_p(a)^T = \sum_{a\in\M}\!\!w(a)\rho_p(a)^{-1} = \sum_{a\in\M}\!\!w(a)\rho_p(a^{-1}) = \sum_{a\in\M}\!\!w(a)\rho_p(a) = \rho_p(\s)\,,\]
where we have used the fact that we may choose a basis such that each $\rho_p(\sigma)$ is an orthogonal matrix for each $\sigma\in\SN$ \cite{sagan}. 
In fact, under any model that is time reversible, the irreducible representations of $\s$ are diagonalisable as, for each $a,a^{-1}\in\M$ such that $a\neq a^{-1}$, the expression for $\rho_p(\s)$ will include a term of the form
\[ w(a)\rho_p(a)+w(a^{-1})\rho_p(a^{-1})=w(a)\left(\rho_p(a)+\rho_p(a^{-1})\right) \]
which, as above, is symmetric (recall that time reversibility demands that $w(a)\!=\!w(a^{-1})$).

Under $(\M_1,w_1,\Pois)$, MLEs were calculated in \cite{mles} for genomes with up to 9 regions.
For both $(\M_1,w_1,\Pois)$ and $(\M_2,w_2,\Pois)$, we were able to calculate MLEs for genomes with 11 regions before hitting the limits of our available computational power.

\begin{example}
We return to the case of genomes with five regions (c.f. Example~\ref{eg S5}).

For the biological model $(\M_1,w_1)$, substituting the exact expressions (\ref{eq S5 pp(e)}) for $\alpha_k(e)$ into the first line of (\ref{eq like}), we obtain:
\[ L(T|e) \;=\; \sum_{k=0}^{\infty} \frac{(5^{2k} + 5)}{6\cdot 5^{2k}} \,\frac{T^{2k} e^{-T}}{(2k)!} \;=\; \tfrac{1}{6}e^{-T}(\cosh(T) + 5\cosh(T/5))\,.
\]
This may be rewritten as
\[ L(T|e) \;=\; \tfrac{1}{12}e^{-T}(e^T + e^{-T} + 5(e^{T/5} + e^{-T/5})) \,,
\]
which, comparing with the last line of (\ref{eq like}), shows that the eigenvalues contributing to the likelihood function in this case are $\pm 1$ and $\pm \tfrac{1}{5}$.
In fact, these are the key eigenvalues for all of the genomes with five regions under this model. 
Under this model, $\mathcal{S}_5$ splits into four equivalence classes: two classes each containing 10 permutations (where the permutations in each class represent one genome) and two classes each containing 50 permutations (where the permutations in each class represent 5 genomes).
The remaining three likelihood functions (displayed for a representative taken from each equivalence class) are
\begin{align*} 
L(T|(1,2,3)) \;&=\; \tfrac{1}{12}e^{-T}(e^T + e^{-T} - (e^{T/5} + e^{-T/5})) \,, \\
 L(T|(1,2)) \;&=\; \tfrac{1}{12}e^{-T}(e^T - e^{-T} + (e^{T/5} - e^{-T/5})) \,,\\
L(T|(1,2,4,3)) \;&=\; \tfrac{1}{12}e^{-T}(e^T - e^{-T} - 5(e^{T/5} - e^{-T/5})) \,. 
\end{align*}

The MLEs for these genomes can be found by differentiating the likelihood functions and numerically solving the resulting quintic equations for the unknown $x= e^{-2T/5}$.
Via this procedure, we found that
\begin{itemize}
\item The genome in the equivalence class with representative $e$ has MLE $\widehat{T}=0$ (as is expected for the reference genome).
\item The 5 genomes in the equivalence class with representative $(1,2)$ each have MLE $\widehat{T}=1.82926$.
\item The 5 genomes in the equivalence class with representative $(1,2,3)$ and the genome in the equivalence class with representative $(1,2,4,3)$ do not have an MLE (since the likelihood function has no maximum\footnote{Further examples and discussion of this phenomenom are given below.}).
\end{itemize}

For each equivalence class, repeating these calculations for the biological model $(\M_2,w_2)$  produced precisely the same results. 
In hindsight, the reason for this is clear: for five regions, the additional rearrangements in $\M_2$, e.g. $(3,5)$, produce, up to dihedral symmetry, precisely the same result as a corresponding rearrangement in $\M_1$, vis-\'a-vis $(1,2)$.
Coupled with our particular choice of rearrangement probabilities $w_2(a)$ (\ref{eq:w2}), we see that there is no difference between our two biological models when considered on genomes of five regions. \hfill $\diamond$
\end{example} 

For each $N\!>\!5$, each of our two models produces at least some eigenvalues $\lambda_{p,i}$ that cannot be given in exact (algebraic) form.
However, as can be seen from the last line of (\ref{eq like}), in practice we can (and do) bypass calculation of the path probabilities $\alpha_k(\sigma)$. 
For a given genome represented by a permutation $\sigma\in\SN$ and each irreducible representation $\rho_p$ and eigenvalue $\lambda_{p,i}$, the focus of the calculation is on the  quantities $\tr(\rho_p(\sigma^{-1} \D) E_{p,i})$. 
For each $\sigma$, we refer to the quantities $\tr(\rho_p(\sigma^{-1} \D) E_{p,i})$ as \emph{partial traces}.

As described in \cite[Sec 3]{jezandpet}, one may compute the projections $E_{p,i}$ via the standard technique
\begin{equation}\label{eq:projbad} 
E_{p,i} = \prod_{j \neq i} \frac{\rho_p(\s) - \lambda_{p,j} I}{\lambda_{p,i} - \lambda_{p,j}} \,.
\end{equation}

For $N\!=\!6$, although several of the eigenvalues were not obtainable in algebraic form, we found that these particular eigenvalues had, for all genomes, corresponding partial trace values of zero. 
That is, the only eigenvalues that contributed to the likelihood function were ones that we could calculate in exact algebraic form; in fact, for $(\M_2,w_2)$, the contributing eigenvalues were all rationals. 
A re-examination of the non-zero partial traces confirmed that these were all also obtainable in algebraic form and, in this way, for each model we obtained exact expressions for the likelihood functions for all genomes in $\mathcal{S}_6$. 

The ten equivalence classes in $\mathcal{S}_6$ comprise one class containing 1 genome, one class containing 2 genomes, three classes each containing 3 genomes, two classes each containing 6 genomes, and three classes each containing 12 genomes. 
For the genome $\sigma = (4,6)$, the likelihood functions for genomes in $[\sigma]_{DR}$, under the models $(\M_1,w_1,\Pois(1))$ and  $(\M_2,w_2,\Pois(1))$ respectively, are
\[\begin{split} 
	L_1(T|\sigma)&= \tfrac{e^{-T}}{1560} \, (5 (\, 3 \, \sqrt{13} - 13) e^{-\frac{1}{6} \, T (\sqrt{13} + 1)} - 5 \, (3 \, \sqrt{13} + 13) e^{\frac{1}{6} \, T (\sqrt{13} - 1)} \\
		&\;\;\;- 39 \, (\sqrt{5} - 3) e^{\frac{1}{6} \, T (\sqrt{5} + 1)} + 39 \, (\sqrt{5} + 3) e^{-\frac{1}{6} \, T (\sqrt{5} - 1)}
		- 130 \, e^{-\frac{1}{3} \, T} + 26 \, e^{T})\,; \\[1ex]
L_2(T|\sigma) &= \tfrac{e^{-T}}{60}  \, \left(4 \, e^{\frac{4}{9}T} + 5 \, e^{\frac{1}{9}T} - 6 \, e^{-\frac{5}{9}T} + e^{T} - 4\right)\,.\end{split} \]
The remaining nine likelihood functions for each model are given in the supplementary material.

Implementing (\ref{eq:projbad}) in SageMath \cite{sage}, we were also able to calculate the required projection operators for $N\!=\!7$. 
However, at $N\!=\!8$, the large number of distinct eigenvalues ($>30$) for some of the irreducible representations (and their close proximity to one another) produced catastrophic cancellation and hence large errors in the resulting projection matrices calculated via (\ref{eq:projbad}).\footnote{The errors were easily identified by, for example, summing the projection matrices for a given irreducible representation.} 

As an alternative to the form (\ref{eq:projbad}), a projection $E$ onto an eigenspace of a symmetric matrix may be expressed in terms of the orthonormal eigenvectors $\{v_j\}$ that span the eigenspace: $E\! =\! \sum_{j=1}^r v_j\, v_j^T\,$. 
Then, for each representation matrix $\rho_p(\s)$ and eigenvalue $\lambda_{p,i}$, there exist orthonormal eigenvectors $\{v_{p,i,j}\,:\, j=1,\ldots r\}$ spanning the eigenspace of $\lambda_{p,i}$.
Via a little linear algebra, one obtains
\begin{equation}\label{eq partial calc}
\tr(\rho_p(\sigma^{-1} \D) E_{p,i}) = \sum_{j=1}^r v_{p,i,j}^T \,\rho_p(\sigma^{-1} \D)\, v_{p,i,j}\,. 
\end{equation}

Accordingly, having used SageMath's \texttt{seminormal} representations of the symmetric group up to $N\!=\!7$ (matrices with rational entries), we switched to \texttt{orthogonal} representations at $N\!=\!8$, with the matrices seated in the \texttt{real double field} to enable efficient computation of eigenvectors (along with eigenvalues) via a standard SageMath function. 
%Note we have substituted conjugate transposes $^\dagger$ for transposes $^T$ in this expression. 

Although symmetric matrices have, in theory, real eigenvalues and eigenvectors, the numerical computation of eigenvalues via this method in SageMath produced a mix of real and complex eigenvalues and eigenvectors (albeit with tiny imaginary part). 
Thus a decision was made, for all $N\!\geq\!8$, to implement the calculations in the complex realm, with eigenvalues and partial traces converted to real numbers at the very last step by taking their real part. Note that (\ref{eq partial calc}) is valid for this more general case: we simply substitute conjugate transposes $^\dagger$ for transposes $^T$ in the expression (\ref{eq partial calc}).% and hermitian matrices for symmetric in the preceding discussion.

Further, eigenvalues computated via this method were much less exact. For example, in $\mathcal{S}_8$, for a $7\times 7$ matrix that we knew to have four distinct integer eigenvalues, SageMath found seven distinct eigenvalues (three pairs of which differed by approximately $10^{-14}$). 
We dealt with this via a binning algorithm which grouped together eigenvalues that we felt ``should have'' been the same. 

This proceeded by sorting the eigenvalues for a given $\rho_p(\s)$, and putting consecutive eigenvalues in the same bin whenever their difference was less than a given tolerance (set at $\tfrac{\;\;10^{-9}}{\!N}$). 
The mean of the eigenvalues in each bin was then taken to be a repeated eigenvalue with the appropriate multiplicity.
Eigenvectors corresponding to the eigenvalues in each bin were also grouped together, and these sets spanning the eigenspaces were orthonormalised using the Gram-Schmidt method to allow the calculation of the partial traces via (\ref{eq partial calc}). 
Figure~\ref{flow} summarises the key steps in our computation structure.

\begin{figure}
\begin{center}
	\includegraphics{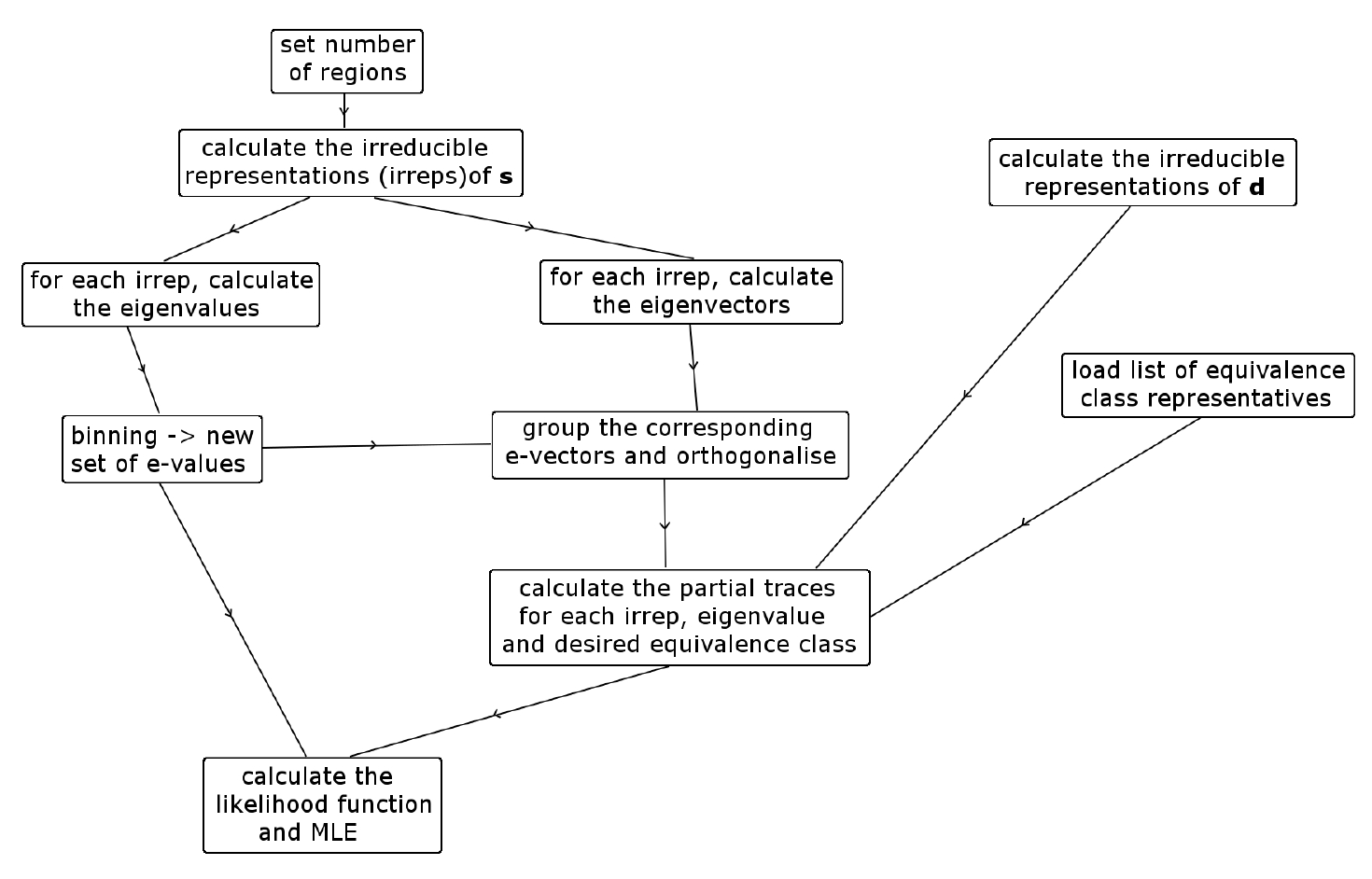}
	\caption{Key steps in computing the MLE. Note that all but the final step were carried out in SageMath; the final calculation of the likelihood function and its maximum were carried out in R.}
	\label{flow}
\end{center}
\end{figure}

Using this approach, we were able to obtain verifiably sensible results for rearrangement models with up to $N\!=\!11$ regions.\footnote{We computed path probabilities $\alpha_k(\sigma)$ both via partial traces (\ref{eq alpha sum 1}) and directly from the irreducible representations (\ref{eq alpha irreps})  --- in the latter case, avoiding eigenvalue/eigenvector estimation --- and these coincide. Additionally, for the cases predicted theoretically by the results of Section~\ref{sec dihedral}, we obtained zero partial trace values (within the expected numerical tolerance).} 
This may not seem a great improvement but one must reflect upon the combinatorial nature of the problem at hand: the $N\!=\!11$ case is $11\cdot 10\cdot 9\cdot 8=79200$ times harder than the $N\!=\!7$ case. 
Further, modulo the approximations necessitated by numerical computation, our results are exact; we have improved on the results obtained in \cite{mles} for $N\!=\!9$ regions without yet introducing any ``on purpose'' approximations. 
In the discussion we will outline future plans for introducing bone-fide numerical approximations.

At $N\!=\!12$, our available RAM was not sufficient for SageMath to compute the irreducible representations of $\s$.
For $\mathcal{S}_{12}$, there are 77 irreducible representations, with a maximum dimension of 7700, compared to 56 irreducible representations, with a maximum dimension of 2376, for $\mathcal{S}_{11}$. 

For all cases up to $N\!=\!9$ we provide, in the supplementary material, a complete set of likelihood function plots and MLEs for each model and every equivalence class of $\mathcal{S}_N$. The proportion of genomes with $N$ regions, up to $N\!=\!9$, that possess an MLE under each of our models is given in Table~\ref{tab:percMLE}. 

\begin{table}[h]
\begin{tabular}{c|ccccc}
$N$ & 5 & 6 & 7 & 8 & 9 \\
  $(\mathcal{M}_1,w_1)$ & 50.0 & 51.7  & 52.8 & 45.8 & 44.6 \\
								 $(\mathcal{M}_2,w_2)$ & 50.0 & 50.0 & 54.7 & 45.4 & 44.0 \\						
\end{tabular}
\medskip
\caption{Percentages of genomes on $N$ regions possessing an MLE.
Results are given for the two biological models $(\mathcal{M}_1,w_1)$ and $(\mathcal{M}_2,w_2)$ described at the start of Section~\ref{sec implement}.}\label{tab:percMLE}
\end{table}

Our proportion of genomes in $\mathcal{S}_9$ with an MLE is slightly higher then the $\sim44\%$ calculated in \cite{mles}. However, given that there are $20 160$ genomes in $\mathcal{S}_9$, even a difference of $0.2\%$ corresponds to around 40 genomes. We have not identified the source of this discrepancy.

For $N\!=\!10$ and $11$ regions, we did not produce results for each of the equivalence classes (numbering $4975$ and $42529$ respectively),
rather calculating MLEs for a sample of genomes in each case. 
Example results for genomes with $N\!=\!10$ regions are presented in Figure~\ref{fig:likes} and Figure~\ref{fig:likes2}.
The plots display the the minimum distance (calculated as the number of rearrangements for which the path probability first attains a non-zero value) as well as the MLE for each case, with the results showing that minimum distance is a poor proxy for true evolutionary time.

Also included in the plots is a measure of the curvature of the likelihood function at the maximum (where this exists).
In order to return interpretable numerical values, we calculated this by taking the negative logarithm of the second derivative of the likelihood function at the maximum.
Given its close relationship to the Fisher information, these values can be interpreted as providing a measure of uncertainty in the MLE.

\begin{figure}%[here]
	\centering
	\scalebox{1.0688}{
		\begin{tabular}{cc}
			\includegraphics[scale=.4]{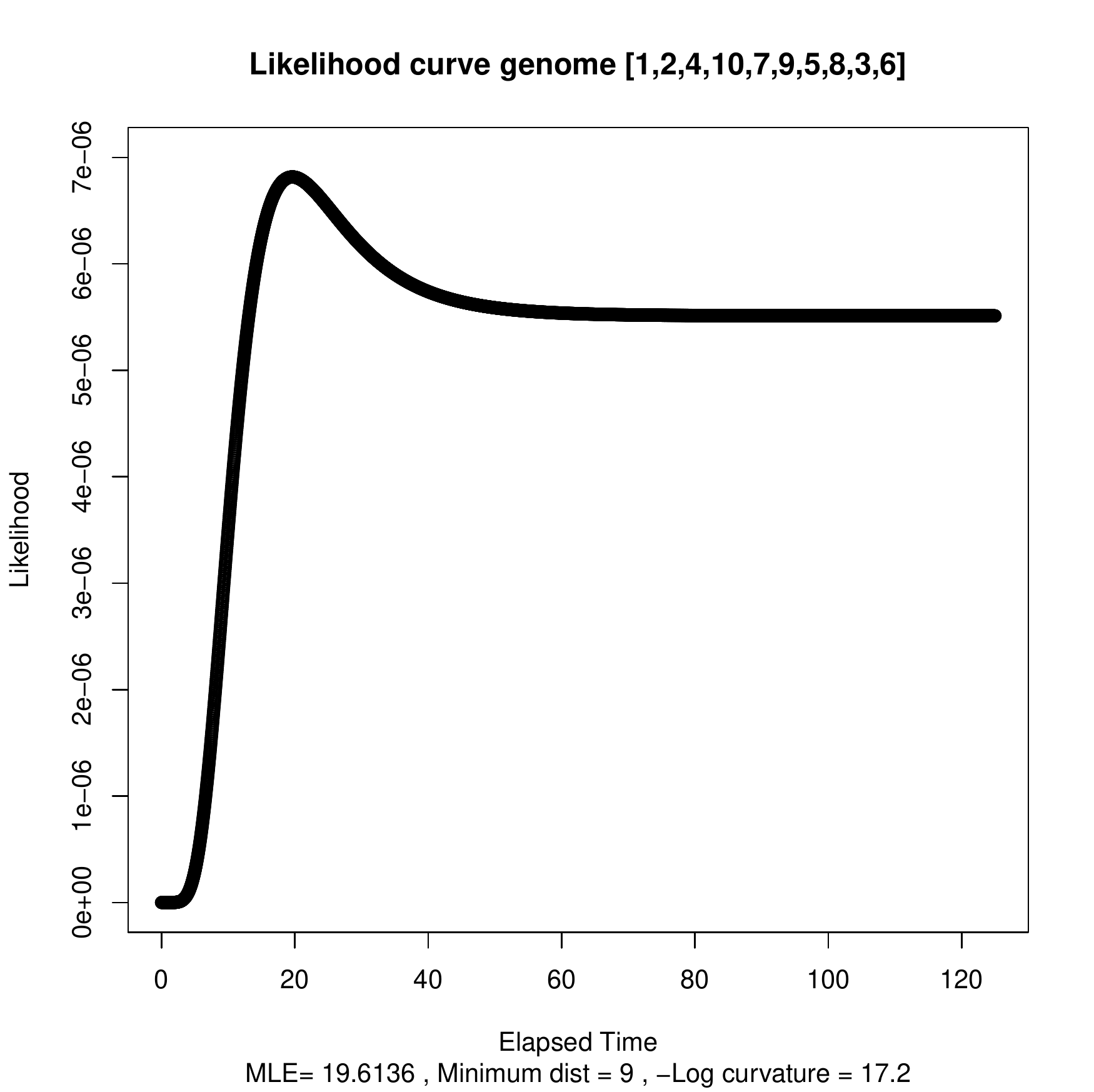} & \includegraphics[scale=.4]{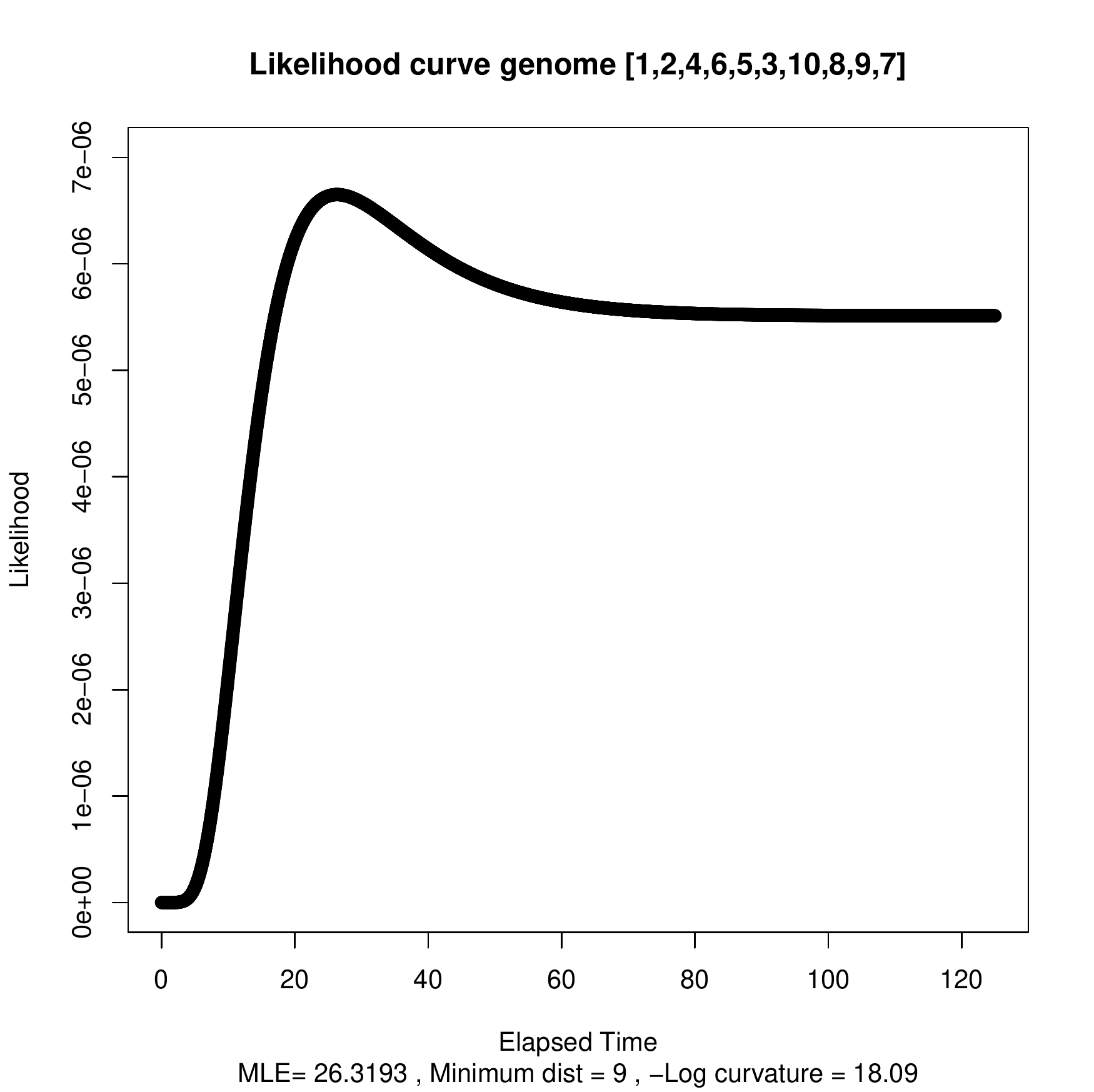}\\ 
			\includegraphics[scale=.4]{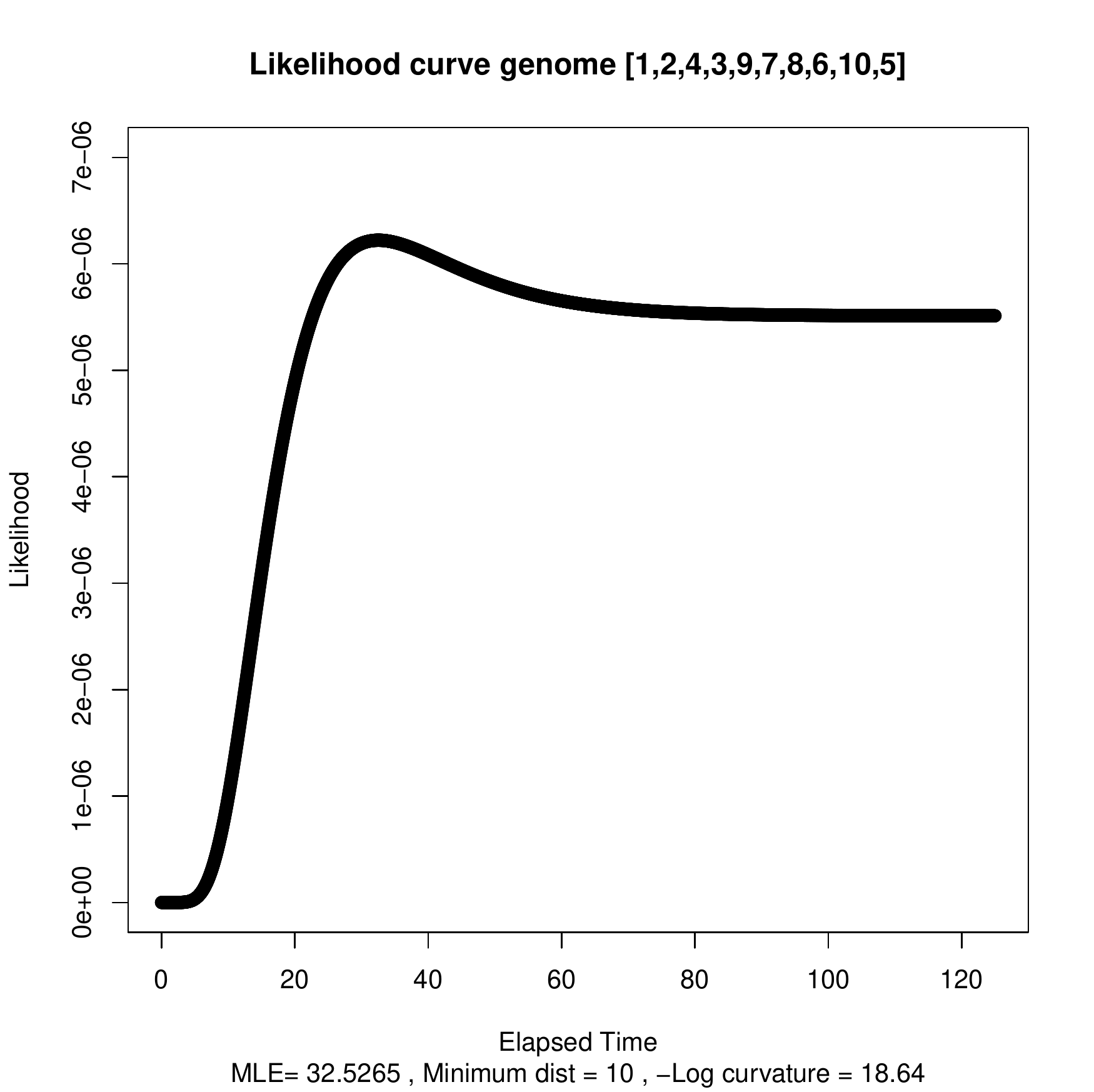} &
			\includegraphics[scale=.4]{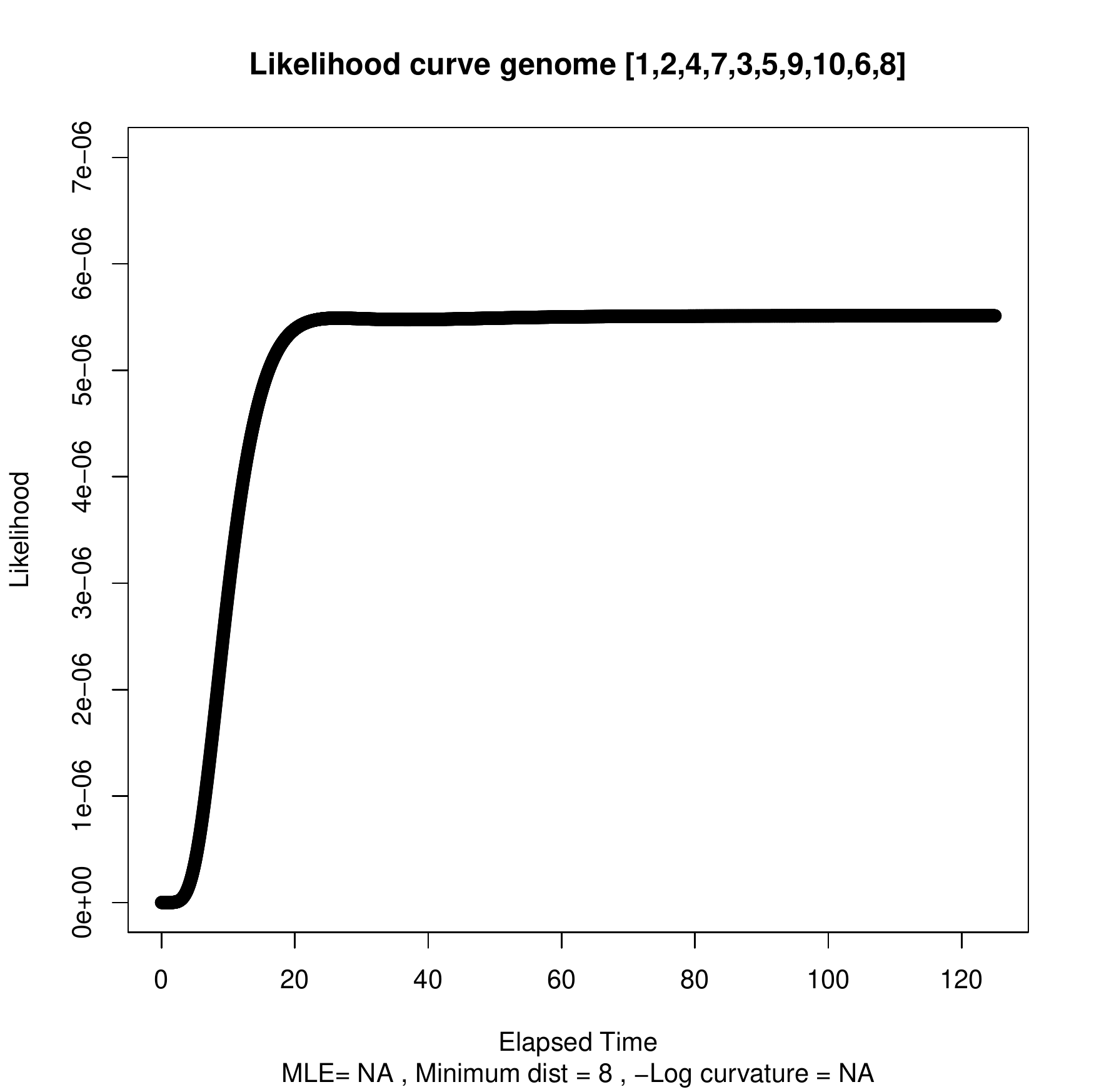}
		\end{tabular}
	}
	\caption{{Likelihood curves of time elapsed for four genomes with ten regions (assuming reference genome $e$ as common ancestor) under $(\mathcal{M}_1,w_1)$.
For ease of comparison, the genomes are represented in one-line notation. For each genome, we have displayed the MLE to the reference, the minimum distance,  and the negative logarithm of the curvature at the MLE.
	%\rt{The maximum minimum distance for all genomes we considered under this model is also indicated.}
	The inconsistencies between the maximum likelihood estimates (MLEs) of elapsed time and minimum distances illustrate the importance of taking a likelihood approach to the computation of evolutionary distance. }}
	\label{fig:likes}
\end{figure}

\begin{figure}%[here]
	\centering
	\scalebox{1.0688}{
		\begin{tabular}{cc}
			\includegraphics[scale=.39]{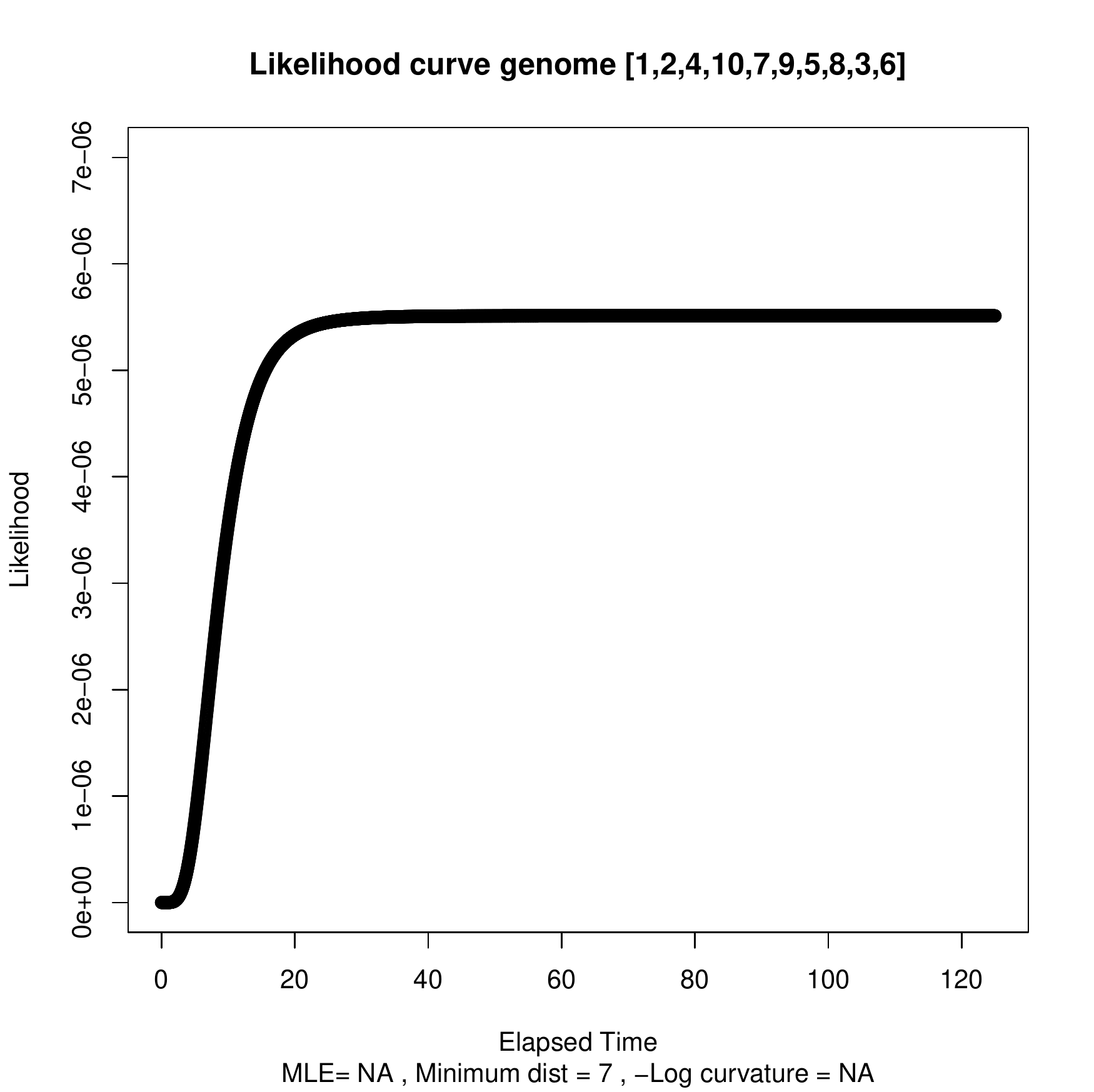} & \includegraphics[scale=.39]{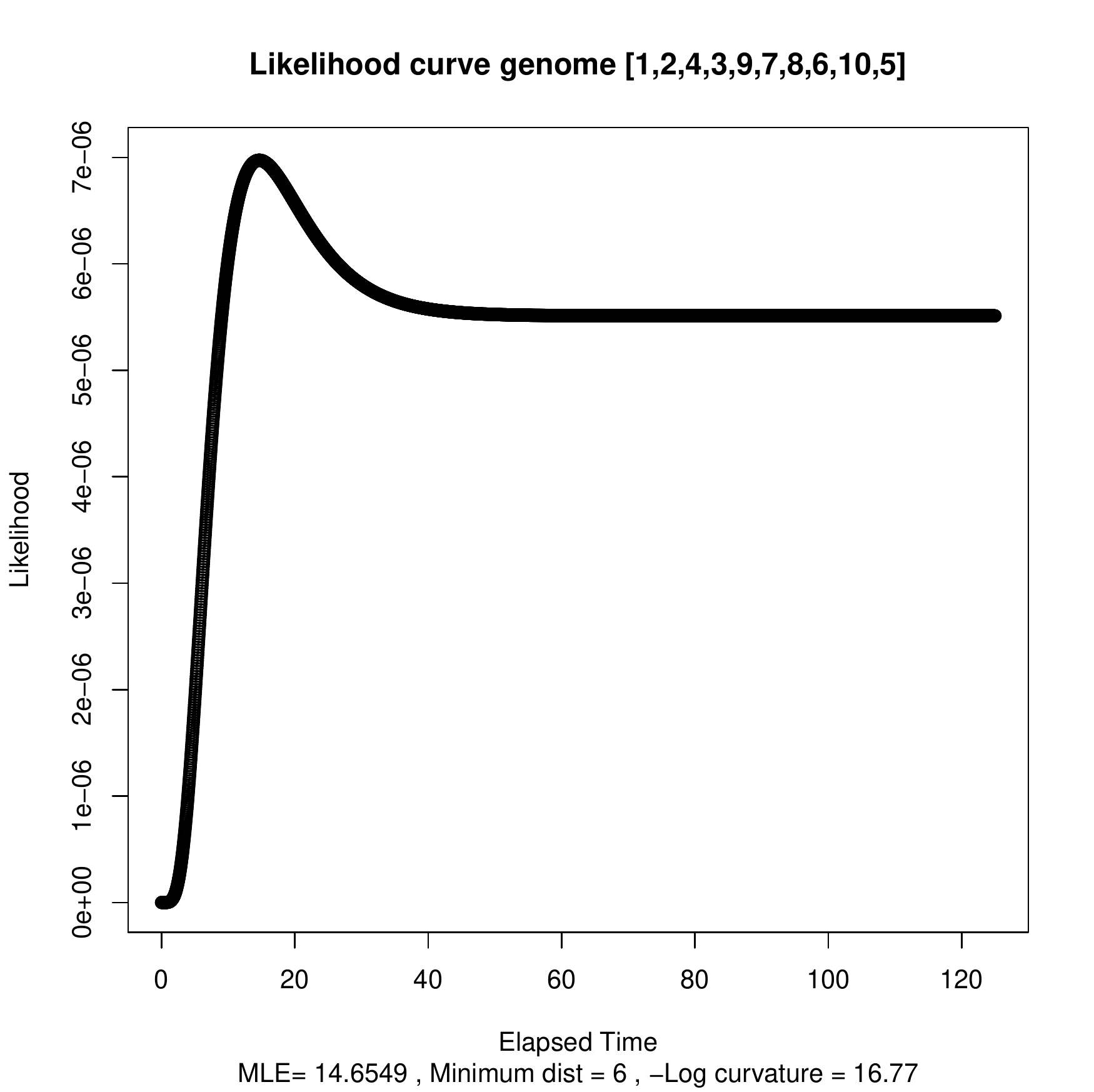}
		\end{tabular}
	}
	\caption{{Likelihood curves of time elapsed for two genomes (c.f. Figure~\ref{fig:likes}) with ten regions (assuming reference genome $e$ as common ancestor) under $(\mathcal{M}_2,w_2)$. }.}
	\label{fig:likes2}
\end{figure}

In particular, note that the MLE clearly distinguishes between the two genomes in Figure \ref{fig:likes} that have a minimal distance of 9. The likelihood curve for the fourth genome does not attain a maximum and hence no reasonable estimate of time elapsed is obtainable.
This information is missed by previous approaches in the field of genome rearrangements which compute minimal distance only (and we ruminate upon this further in the discussion).

Our claim that the maximum likelihood approach provides a much more refined estimate of evolutionary distance than the minimum distance is further supported by an examination of the range of values taken by each distance measure.
In $\mathcal{S}_9$, under each of our models, there are 686 equivalence classes and the minimum distance for any genome under the model $(\M_1, w_1)$ is an integer value between $0$ and $11$. 
Of the equivalence classes, 318 possess an MLE, and each of these 318 MLEs {\em is a distinct value} between 0 and 66.07. Figure~\ref{fig:mlemins} shows the distribution of these MLEs over the minimum distances for $\mathcal{S}_9$. The minimum distance for classes with no MLE ranged between 6 and 11.

\begin{figure}
	\centering
	\scalebox{.75}{
			\includegraphics[scale=.8]{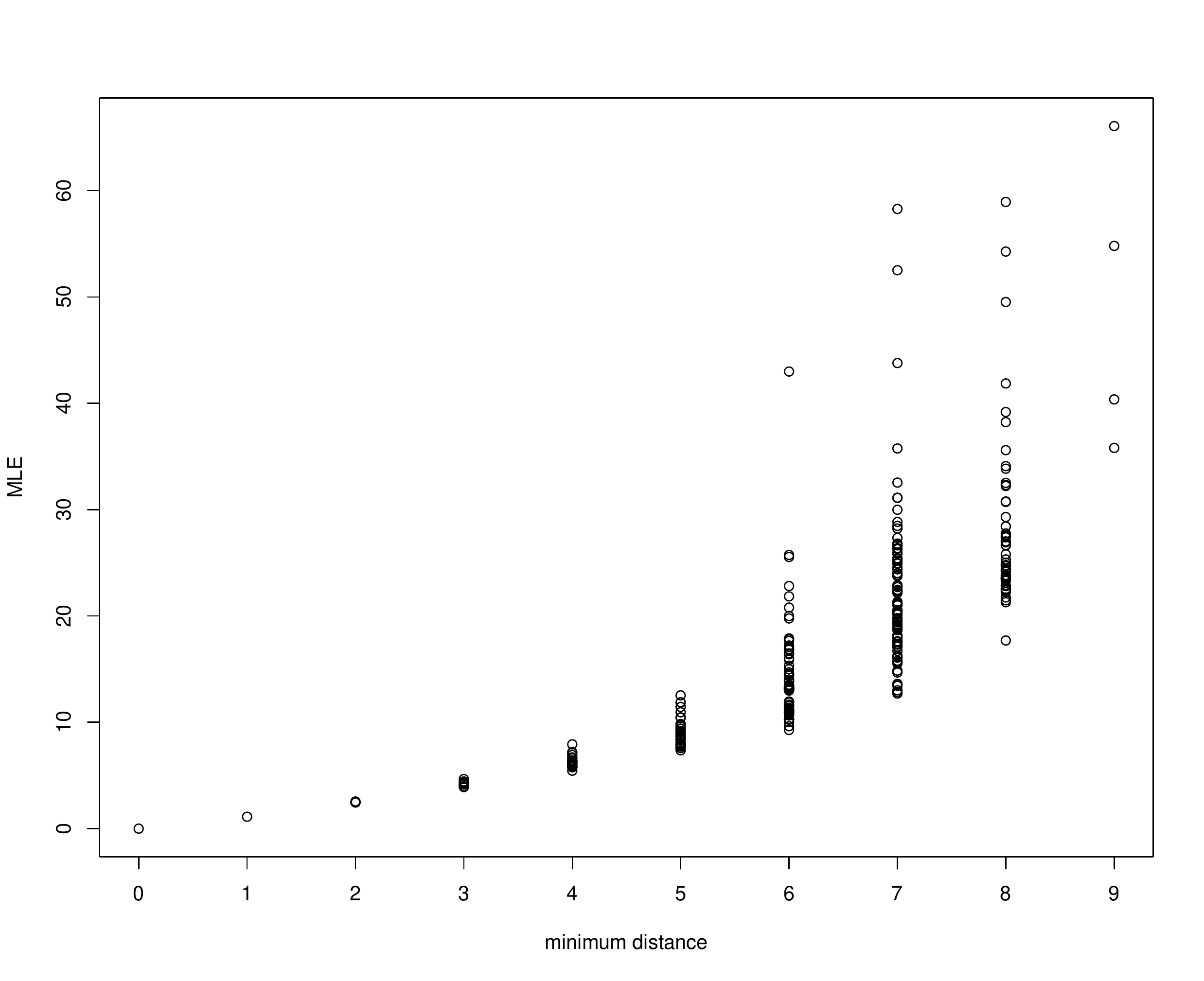}	}
	\caption{A plot of MLEs against mimimum distance for equivalence classes in $\mathcal{S}_9$ possessing an MLE under $(\M_1,w_1)$.}
	\label{fig:mlemins}
\end{figure}

In $\mathcal{S}_9$ under model $(\M_2, w_2)$, the minimum distance took a maximum value of 7 and the MLE a maximum value of 30.46. 
Again, the 316 MLE values obtained in this case were all distinct. 
As a larger rearrangement model will generally result in a reduced range of minimum distances, we again see that \emph{MLE distances give an enormous increase in distinguishability of circular genomes on the basis of their evolutionary relatedness}.

\section{Discussion}

In this work we set out to explore the practical issues that arise when applying representation theory to evolutionary distances calculated under rearrangement models in a maximum likelihood framework.
Using open source computational software and a moderate amount of computing power, we showed that it is plausible to compute,  under general models of rearrangement, maximum likelihood estimates of elapsed evolutionary time for circular genomes with up to eleven regions. 
We worked with two particular models but our results should be replicable for any choice of rearrangements and probability distribution, $(\M,w)$. 

In seeking to avoid unnecessary repetition of calculations, we explored the symmetry of rearrangement models, and found three types of equivalence of permutations in $\SN$, depending on three levels of symmetry that a model may exhibit. 
We defined two permutations to be equivalent under a given model if their likelihood functions coincide. This formulation is equivalent to all path probabilities $\alpha_k$ for the two genomes being equal. Thus the equivalence classes are not dependent on the distance measure being an MLE and remain the same for any alternative measure of distance that is solely based on path probabilities (equivalently, path counts in the case of a uniform probability distribution on $\M$).

We reflect on some some interesting properties of rearrangement models in a likelihood context that are illustrated by our results.
First of all, one may question why around half of the genomes for any number of regions under each model considered were found not to have MLEs --- what does this mean? Is this a feature or a fault of the maximum likelihood approach? 
A precedent that helps us resolve these questions is given nicely by the so-called ``Jukes-Cantor distance correction'' for pairwise DNA alignment distances.
As is well known (see, for example, \cite{felsenstein2004inferring}), under the Jukes-Cantor model of DNA substitution, we may calculate the maximum likelihood estimate of time elapsed as an analytic function of the Hamming distance $\Delta$ and sequence length $r$: 
\[
\widehat{T}=-\tfrac{3}{4}\log(1-\textstyle{\frac{4}{3}}\frac{\Delta}{r}).
\]
Clearly this formula is valid if and only if $\frac{\Delta}{r}\!<\!\tfrac{3}{4}$.
The critical value $\frac{\Delta}{r}\!=\!\tfrac{3}{4}$ corresponds to the case where, relative to sequence length, sufficient time has passed such that the difference between the two sequences is indistinguishable from random noise.
That is, under the model, $\frac{\Delta}{r}\to \tfrac{3}{4}$ as $T\to \infty$ and we can say that these sequences are at ``saturation/equilibrium'' (with respect to the model).

Examining the plots in Figures~\ref{fig:likes} and \ref{fig:likes2}, it is clear that the likelihood curves for genomes $[1,2,4,7,3,5,9,10,6,8]$ and $[1,2,4,10,7,9,5,8,3,6]$ (respectively) are strictly increasing and hence there is no maximum likelihood estimate of elapsed time distinguishable from $\widehat{T}\to \infty$. 
In this situation, we see that, under the respective models, these genomes are at saturation with respect to the reference genome $e$.
The practical consequence of this is that the evolutionary distance from the reference is simply unobtainable for these genomes --- under this modelling scenario, any chance of recovering evolutionary signal has either been lost or was \emph{non-existent in the first place}.
This information is not available using minimal distances as, given that the model set generates the whole of $\SN$, every genome is some minimal distance from every other genome. 
We are of the opinion that this further illustrates the limitations of taking minimal distance as a proxy for true evolutionary distance in the context of rearrangement models.

Further, the information regarding the \emph{uncertainty} of the estimate (via curvature in the likelihood curve around the optimum) presents the opportunity develop a more refined distance-based clustering method that takes into account this uncertainty.

From our current position, we can see several potentially fertile avenues that as yet remain unexplored. 
The most compelling is the introduction of {\em intentional} numerical approximation (as opposed to those necessitated by computing with finite accuracy). This was proposed in \cite{jezandpet}, and has not yet been attempted in a serious fashion. 
It was observed in \cite{jezandpet} that the largest eigenvalues contribute the most to the likelihood, however, with all eigenvalues occurring in the interval $[-1,1]$, and already around 5000 eigenvalues in play over all irreducible representations in $\mathcal{S}_{10}$, trying to distinguish between ``large'' and ``smaller'' eigenvalues may not be the best way forward. 

Having already introduced binning of eigenvalues (as a means of improving accuracy), we intend to explore coarser binning as a means of numerical approximation. This would reduce the number of eigenvalues, and thus the number of terms in the likelihood function, but as the total dimension of the eigenspaces for each irreducible representation does not change, the calculation of the partial traces would still be as computationally intensive, unless we were able to reduce the number of eigenvalues to a level at which we could return to calculating projections via the original method (\ref{eq:projbad}) and thus avoid eigenvectors entirely.  

In a different direction, one sees from the exact expressions for the likelihood functions in $\mathcal{S}_5$ and $\mathcal{S}_6$, and from an inspection of the partial trace values for larger numbers of regions, that the majority of eigenvalues occuring in the theoretical likelihood function do not in fact contribute to the final likelihood function, as their coefficient is zero. If one could determine which eigenvalues will contribute, (or predict which partial traces will be zero {\em without needing to calculate them} --- as we did in Section~\ref{sec dihedral} for whole irreducible representations), the calculation load would be vastly reduced without needing to resort to numerical approximation. 

We leave these possibilities for future work.

\bibliographystyle{plain}
\bibliography{biblio}

\end{document}